\documentclass[12pt]{amsart}

\usepackage{latexsym,amssymb,amsfonts,amsmath,mathrsfs,float,hyperref, tikz}
\usepackage{thmtools,thm-restate}
\usepackage{xcolor}
\usepackage{tikz}
\usetikzlibrary{quantikz}
\hypersetup{
    colorlinks,
    linkcolor={black},
    citecolor={black},
}

\calclayout

  \renewcommand{\Pr}{\mbox{\rm Pr}}
  \newcommand{\Exp}{{\mathbb{E}}}

  \newcommand{\C}{\mathbb{C}} 
  \newcommand{\N}{\mathbb{N}} 
  \newcommand{\F}{\mathbb{F}} 
  \DeclareMathOperator{\im}{im} 

  \DeclareMathOperator{\isbal}{IsBal}
  \DeclareMathOperator{\maj}{Maj}
  \DeclareMathOperator{\LH}{LH}

  \newcommand{\eg}{{e.g.}} 
  \newcommand{\eps}{\varepsilon}

  \newcommand{\poly}{\mbox{\rm poly}}
  \newcommand{\ceil}[1]{\lceil{#1}\rceil}

  \DeclareMathOperator{\rank}{rk}

  \DeclareMathOperator{\arank}{arank}
  \DeclareMathOperator{\prank}{prank} 
  \DeclareMathOperator{\bias}{bias}
  \DeclareMathOperator{\Pol}{Pol}
   
  \newcommand{\beq}{\begin{equation}}
  \newcommand{\eeq}{\end{equation}}
  \newcommand{\beqn}{\begin{equation*}}
  \newcommand{\eeqn}{\end{equation*}}
  \newcommand{\beqr}{\begin{eqnarray}}
  \newcommand{\eeqr}{\end{eqnarray}}
  \newcommand{\beqrn}{\begin{eqnarray*}}
  \newcommand{\eeqrn}{\end{eqnarray*}}
  \newcommand{\bmline}{\begin{multline}}
  \newcommand{\emline}{\end{multline}}
  \newcommand{\bmlinen}{\begin{multline*}}
  \newcommand{\emlinen}{\end{multline*}}
  
  \newcommand{\nc}{{\normalfont NC}}
    \newcommand{\ncz}{{\normalfont NC}$^0$}
  \newcommand{\qnc}{{\normalfont QNC}}
  \newcommand{\ac}{{\normalfont  AC}}
    \newcommand{\acz}{{\normalfont  AC}$^0$}
  
  \newcommand{\qac}{{\normalfont  QAC}}

  \theoremstyle{plain}
  \newtheorem{theorem}{Theorem}[section]
  \newtheorem{lemma}[theorem]{Lemma}

  \newtheorem{corollary}[theorem]{Corollary}
  
  \theoremstyle{definition}
  \newtheorem{definition}[theorem]{Definition}

  \theoremstyle{remark}
  \newtheorem{remark}[theorem]{Remark}
  
  \renewenvironment{proof}[1][]{
    	\begin{trivlist}
     	\item[\hspace{\labelsep}{\em\noindent Proof#1:\/}]}
     	{{\hfill$\Box$}
    	\end{trivlist}
  }

\begin{document}
\title[Noisy decoding by shallow circuits with parities]{Noisy decoding by shallow circuits with parities: classical and quantum}

\author{Jop Bri\"{e}t}
\address{CWI \& QuSoft, Science Park 123, 1098 XG Amsterdam, The Netherlands}
\email{j.briet@cwi.nl}

\author{Harry Buhrman}
\address{QuSoft \& University of Amsterdam \& CWI, Science Park 123, 1098 XG Amsterdam, The Netherlands}
\email{buhrman@cwi.nl}

\author{Davi Castro-Silva}
\address{CWI \& QuSoft, Science Park 123, 1098 XG Amsterdam, The Netherlands}
\email{davi.silva@cwi.nl}

\author{Niels M. P. Neumann}
\address{CWI \& QuSoft, Science Park 123, 1098 XG Amsterdam, The Netherlands, \& The Netherlands Organisation for Applied Scientific Research (TNO)}
\email{niels.neumann@tno.nl}

\thanks{This work was supported by the Dutch Research Council (NWO/OCW), as part of the Quantum Software Consortium programme (project number 024.003.037),  the NETWORKS programme (grant no. 024.002.003), and the Quantum Delta NL program.\\ This is the full version of an extended abstract that appeared in the proceedings of ITCS'24.}

\maketitle

\begin{abstract}
We consider the problem of decoding corrupted error correcting codes with \nc$^0[\oplus]$ circuits in the classical and quantum settings.
We show that any such classical circuit can correctly recover only a vanishingly small fraction of messages, if the codewords are sent over a noisy channel with positive error rate.
Previously this was known only for linear codes with large dual distance,
whereas our result applies to any code.
By contrast, we give a simple quantum circuit that correctly decodes the Hadamard code with probability~$\Omega(\eps^2)$ even if a $(1/2 - \eps)$-fraction of a codeword is adversarially corrupted.

Our classical hardness result is based on an equidistribution phenomenon for multivariate polynomials over a finite field under biased input-distributions. 
This is proved using a structure-versus-randomness strategy based on a new notion of rank for high-dimensional polynomial maps that may be of independent interest.

Our quantum circuit is inspired by a non-local version of the Bernstein-Vazirani problem, a technique to generate ``poor man's cat states'' by Watts et al., and a constant-depth quantum circuit for the OR function by Takahashi and Tani.
\end{abstract}

\section{Introduction}

Error correcting codes (ECCs), formally introduced in Shannon's celebrated work~\cite{Shannon:1948}, protect digital signals from noise.
An ECC is a map $E:\Sigma^k\to\Sigma^n$, for a finite alphabet~$\Sigma$ and positive integers $n\geq k$, with the property that any message $x\in \Sigma^k$ can be decoded from the codeword~$E(x)$ even if the codeword is partially corrupted.
If too many errors occur, however, recovering the original message may become impossible.
In such cases one can instead resort to \emph{list decoding}, an influential idea proposed in seminal works of Elias~\cite{Elias:1957} and Wozencraft~\cite{Wozencraft:1958}, which aims to give a small list of messages whose codewords are close to the received (corrupted) codeword.
Complexity considerations appear naturally in this context, as encoding and decoding ideally allow for reliable communication with limited computational resources;
they also appear because of the fundamental role played by ECCs in computational complexity itself (see \eg,~\cite{Trevisan:2004} for a survey).

\subsection{Error models}\label{sec:errormodels}
In the error model considered by Shannon~\cite{Shannon:1948}, a codeword is corrupted according to some random process.
A natural such process is given by the \emph{symmetric channel}:
for each coordinate of the codeword independently, the channel either transmits it unchanged with some probability~$\rho$, or replaces it with a uniformly random element of~$\Sigma$ with probability~$1-\rho$.
We refer to~$\rho$ as the \emph{bias} of the channel.\footnote{In this model, each coordinate is thus corrupted with probability $(1 - \rho)(1 - |\Sigma|^{-1})$, which is usually referred to as the \emph{error rate}. For our purposes, however, the bias will be a more convenient parameterization.}
If~$Z\in \Sigma^n$ is distributed according to the random outcome of the symmetric channel with bias $\rho$ applied to a codeword~$E(x)$, we write $Z\sim\mathcal N_\rho\big(E(x)\big)$.
In this model the goal is to correctly decode a corrupted codeword with good probability over the noise. 

The combinatorial worst-case error model of Hamming~\cite{Hamming:1950} instead assumes that the codeword is corrupted arbitrarily on at most some $\delta\in [0,1)$ fraction of coordinates.
We will refer to~$\delta$ as the \emph{error parameter}.
In this setting, the number of errors that can be tolerated depends on the minimal Hamming distance between any pair of distinct codewords, or \emph{minimal distance} of the code, denoted~$d_E$.
Since the Hamming ball of diameter $d_E-1$ around any point $y\in \Sigma^n$ contains at most one codeword, a message can be retrieved if fewer than~$d_E/2$ errors have occurred.

If more errors occur, faithful decoding is no longer possible and list decoding enters the picture. 
For $\delta\in [0,1)$ and positive integer~$L$, a code is \emph{$(\delta,L)$-list decodable} if for any point $y\in \Sigma^n$, the Hamming ball of radius~$\delta n$ centered around~$y$ contains at most~$L$ codewords.
It is well known that any $(\delta,L)$-list decodable code satisfies $L \geq \Omega(1/\eps^2)$ when $\delta = (1-\eps) (1 - |\Sigma|^{-1})$~\cite{Guruswami:2010}.
If fewer than a $\delta$-fraction of codeword coordinates are corrupted, then a random element from this list will give the correct message with probability at least~$1/L$.

\subsection{Circuits}\label{sec:circuits}
A well-studied problem is that of decoding corrupted ECCs by constant-depth circuits with~$n$ inputs,~$k$ outputs and size $\poly(n)$, for example in the context of black-box hardness amplification~\cite{SudanTV:1999, TrevisanVadhan:2007, viola:2006}.
Two classes of such circuits are~\acz{}, consisting of unbounded-fan-in AND, OR and NOT gates, and the class~\ncz{}, consisting of arbitrary bounded-fan-in gates;
without loss of generality, we may assume that the fan-in of any gate in~\ncz{} is at most two.

The extensions of these classes where unbounded-fan-in parity gates are added to the gate sets are denoted by~\ac$^0[\oplus]$ and~\nc$^0[\oplus]$, respectively.
These are proper extensions since parity cannot be computed by~\acz{} circuits and~\ncz{} is a proper subset of~\acz{} (see~\cite{AroraBarak:2009}).
An important distinction is that the outputs of~\ncz{} circuits depend on only a constant number of coordinates of the input, whereas the outputs of~\nc$^0[\oplus]$ circuits can depend on the whole input.
The classes \acz{} and~\nc$^0[\oplus]$ are incomparable since \nc$^0[\oplus]$ cannot compute the $n$-bit AND function;
indeed, \nc$^{0}[\oplus]$ circuits can compute only constant-degree polynomials over~$\F_2$ (see Section~\ref{sec:techniques}), whereas AND has degree~$n$.

We also consider the quantum counterparts of the above circuit classes, denoted QX, where~X is one of the classes discussed above;
these classes were first introduced by Moore~\cite{Moore:1999} and Moore and Nilsson~\cite{MooreNilsson:2001}.
Thus,~\qnc$^0$ is the class of constant-depth quantum circuits containing arbitrary one- and two-qubit gates, while \qnc$^0[\oplus]$ includes unbounded-fan-in parity gates acting on superpositions.
In contrast with their classical analogues, the classes  \qnc$^0[\oplus]$ and \qac$^0$ are known to be equivalent~\cite{GreenHomerMoorePollett:2002,HoyerSpalek:2005,Moore:1999}.\footnote{In addition, the works of Moore showed that these classes are all equivalent to \qac$^0[q]$, the class \qac$^0$ with additional modulo-$q$ gates. 
For an integer $q>1$, a modulo-$q$ gate evaluates to $1$ if the sum of its inputs equals $0\bmod q$ and evaluates to~$0$ otherwise.
Classically, the classes \ac$^0[p]$ and \ac$^0[q]$ are incomparable if~$p$ and~$q$ are powers of distinct primes~\cite{Razborov:1987,Smolensky:1987}.
}
In the setting we consider, all parity gates will be classical.
This is in slight contrast with the works above, which consider quantum parity gates.
The above-mentioned equivalence still holds however, due to the fact that quantum parity can be computed by a \qnc$^0[\oplus]$ circuit (with classical parities).

\subsection{Quantum advantage}
The above-mentioned classes of quantum circuits recently enjoyed renewed interest in the context of provable separations between quantum and classical complexity classes.

One of the principal challenges in quantum computing is to determine for which types of problems quantum computers offer a significant advantage over classical ones.
Celebrated examples of practical importance, such as Shor's algorithm for integer factoring~\cite{Shor:1997}, require quantum computers of a vastly larger scale than currently available. 
Moreover, formally proving classical hardness of factoring appears to be beyond the scope of currently-available techniques.
Constant-depth circuits form an attractive computational model, as they will likely be easier to implement in practice and, from the perspective of complexity theory, provide one of the few settings currently amenable to provable lower bounds.

A recent series of works, starting with a breakthrough of Bravyi, Gosset and K\"{o}nig~\cite{BravyiGossetKoenig:2018}, considered the relative power of \emph{shallow quantum circuits}. For instance:

\begin{itemize}
    \item The 2D-Hidden Linear Function problem can be solved exactly in \qnc$^0$ while any \ac$^0$ circuit succeeds with exponentially small probability under a certain input distribution~\cite{WattsKothariSchaefferTal:2019}; this strengthened the main result of~\cite{BravyiGossetKoenig:2018} showing that this problem separates \qnc$^0$ from \nc$^0$ in the worst case. 
    \item The Relaxed Parity Halving problem can be solved exactly in \qnc$^0$ while any \ac$^0$ circuit succeeds with probability at most $\frac{1}{2} + \exp(-n^\eps)$ under the uniform distribution~\cite{WattsKothariSchaefferTal:2019}.
    \item The Parallel Parity Bending problem can be solved with probability $1 - o(1)$ by a \qnc$^0/\mathsf{qpoly}$ circuit while any  \ac$^0[\oplus]/\mathsf{rpoly}$ succeeds with probability at most $O(n^{-\eps})$~\cite{WattsKothariSchaefferTal:2019}.
    \item The problem of simulating correlations obtained from measuring graph states \qnc$^0$ and \nc$^0$, even in the average-case~\cite{LeGall:2019}.
    \item The 1D-Magic Square problem separates noisy \qnc$^0$ circuits from \nc$^0$~\cite{BravyiGKT:2020}.
\end{itemize}
Similar separations based on other relational and sampling-based problems were proven in~\cite{CoudronSV:2021, GrierSchaeffer:2020, watts2023unconditional}.
A common feature of all these problems is that they were specifically designed to prove separations between shallow quantum and classical circuits.

We instead consider the problem of decoding a corrupted error-correcting code, which arises naturally in computer science.
This problem is well studied in the context of classical complexity theory, where shallow circuits endowed with parity gates are also considered;
see Sections~\ref{sec:qresults} and~\ref{sec:literature} for further discussion.

\subsection{The Hadamard code}\label{sec:Hadamard}
A basic but important example of an ECC is the \emph{Hadamard code}, which encodes $k$-bit messages into codewords of length $n=2^k$ and is given by the $\F_2$-linear map $H(x) = (\langle x,y\rangle)_{y\in \F_2^k}$, where $\langle x,y\rangle = y^\mathsf{T}x$.
This code has minimal distance~$n/2$ and is $(1/2-\eps, O(1/\eps^2))$-list decodable for any $\eps\in (0, 1/2]$, which is known to be optimal for any code~\cite{Guruswami:2010}.

Under the symmetric channel, the Chernoff bound implies that unique decoding of the Hadamard code is possible with high probability for any constant bias $\rho > 0$.\footnote{This even holds for any code over a large enough alphabet, as shown in~\cite{RudraUurtamo:2010}.} 
This is due to the fact that, with high probability, the Hamming ball of radius $(1/4 - \rho/4)n$ around a corrupted version of a codeword~$C$ contains no other codewords than~$C$ itself.

For the worst-case Hamming model, Goldreich and Levin~\cite{GolreichLevin} famously gave an efficient list decoding algorithm for the Hadamard code that runs in time $\poly(k, 1/\eps)$, for error parameter $\delta = 1/2 - \eps$. 
For fixed $\eps>0$, their algorithm gives a probabilistic \ac$^0$ circuit that, on input length $n$, correctly returns the original message with probability~$\Omega(1)$.

\section{Our results}
Here we consider the following problem.
Let $E:\F_2^k\to \F_2^n$ be a (binary) error correcting code.
Given a map $\phi:\F_2^n\to \F_2^k$ representing some decoding procedure, we wish to bound the probability of correct message retrieval:
\beq\label{eq:success}
\Pr\big[\phi\big(E(x) + Z\big) = x],
\eeq
where~$x\in \F_2^k$ is some message and~$Z\in \F_2^n$ is an error string.
We consider two scenarios, one classical and one quantum.

\subsection{Classical setting}
\label{sec:cresults}

In the first scenario,~$\phi$ represents an~\nc$^0[\oplus]$ circuit,~$x$ is uniformly distributed and~$Z\sim\mathcal N_\rho(0)$, so that $E(x) +Z$ is a random codeword corrupted according to the binary symmetric channel with bias~$\rho$.
Our main result in this setting says that~\eqref{eq:success} tends to zero, for any $\rho\in [0,1)$ and any code:

\begin{theorem}[Impossibility of decoding by {\nc$^0[\oplus]$}]\label{thm:nc0plus}
For any $\rho\in [0,1)$, $d\in \N$ and $\eps\in (0,1]$, there is a $k_0 = k_0(d, \rho, \eps)\in \N$ such that the following holds.
Let $k \geq k_0$ and~$n$ be positive integers, $E:\F_2^k\to \F_2^n$ be any map and $\phi:\F_2^n\to\F_2^k$ be a map computable by an \nc$^0[\oplus]$ circuit of depth at most~$d$.
Then, for a uniformly distributed~$x\in \F_2^k$ and $Z\sim\mathcal N_\rho(0)$, we have that
\beqn
\Pr\big[\phi\big(E(x) + Z\big) = x] < \eps.
\eeqn
\end{theorem}

In particular, this theorem shows that no \nc$^0[\oplus]$ circuit can correctly decode more than an $\eps$-fraction of codewords with probability higher than $\eps$ over the noise distribution, if the messages are long enough depending on $\eps$, the error rate $(1-\rho)/2 >0$ and the depth of the circuit.
As a consequence of Yao's minimax principle~\cite{Yao:1997} and the Chernoff bound, it follows that any probabilistic \nc$^0[\oplus]$ circuit will also fail (with high probability) to correctly decode any binary ECC in the worst-case Hamming model, for any constant error parameter~$\delta \in (0, 1/2]$.

\medskip

We note that the decay we obtain on the probability~\eqref{eq:success} of correct message retrieval as a function of the message length is extremely slow, making Theorem~\ref{thm:nc0plus} a qualitative result rather than quantitative.
Nevertheless, we conjecture that the true decay of this probability is exponential in the message length~$k$;
this would clearly be optimal, as can be seen by taking a constant map~$\phi$ which always returns some fixed message.
In Section~\ref{sec:high_char} we will provide some evidence to support this conjecture.

\subsection{Quantum setting}
\label{sec:qresults}

In the second scenario, we consider the worst-case Hamming model with constant-depth quantum circuits.
Our main result in this setting is an explicit \qnc$^0[\oplus]$ circuit capable of decoding the Hadamard code.

\begin{theorem}[Decoding Hadamard with {\qnc$^0[\oplus]$}]\label{thm:qcircuit}
There is a family of \qnc$^0[\oplus]$ circuits $(\mathcal C_n)_{n\in \N}$ such that the following holds.
Let $k\in \N$, $n = 2^k$ and $\eps \in (0,1/2]$.
Then, for any $y\in \F_2^n$ and any $x\in \F_2^k$ satisfying $d\big(y, H(x)\big)\leq (\frac{1}{2} - \eps)n$, on input~$y$ the circuit~$\mathcal C_n$ returns~$x$ with probability~$\Omega(\eps^2)$.
\end{theorem}

We note that the bound~$\Omega(\eps^2)$ obtained in the theorem is optimal, since in general there can be $\Theta(\eps^{-2})$ messages~$x\in \F_2^k$ satisfying $d\big(y, H(x)\big)\leq (\frac{1}{2} - \eps)n$.
This bound is non-trivial only when $\eps = \Omega(1/\sqrt{n})$, as there are~$n$ possible messages.

As a simple corollary of Theorem~\ref{thm:qcircuit}, we obtain a similar result for the problem of list decoding the Hadamard code.

\begin{restatable}{corollary}{qlisthad}
\label{cor:qlisthad}
There is a family of \qnc$^0[\oplus]$ circuits $(\mathcal C_n)_{n\in \N}$ such that the following holds.
Let $k\in \N$, $n = 2^k$ and $\eps \in [1/\sqrt n,1/2]$.
Then, on any input~$y\in \F_2^n$, with probability $1 - \eps$ the circuit~$\mathcal C_n$ returns a list~$L(y)$ of size $O(\eps^{-2}\log (1/\eps))$ which contains every $x\in \F_2^k$ with $d\big(y, H(x)\big)\leq (\frac{1}{2} - \eps)n$.
\end{restatable}

\begin{proof}
For a large enough constant~$C>0$, consider $C\eps^{-2}\log (1/\eps)$ parallel instances of the circuit from Theorem~\ref{thm:qcircuit}.
This gives a list~$L(y)$ of the claimed size such that any message $x\in \F_2^k$ satisfying $d\big(y, H(x)\big)\leq (\frac{1}{2} - \eps)n$ appears in~$L(y)$ with probability at least $1 - \eps^3$.
Since there are at most $O(1/\eps^2)$ such messages, it follows from the union bound that with probability at least~$1-O(\eps)$ every such message appears in~$L(y)$.
\end{proof}

\begin{remark}
Note that the circuits obtained in this corollary also output several messages whose codewords differ from the input $y$ in more than $(\frac{1}{2} - \eps)n$ coordinates;
this differs from the usual notion of the list decoding problem, which aims to output a list of all messages $x\in \F_2^k$ with $d\big(y, H(x)\big)\leq (\frac{1}{2} - \eps)n$ and none other.
One can also solve the usual list decoding problem for the Hadamard code using \qnc$^0[\oplus]$ circuits, by making use of MAJORITY gates (and more general threshold gates) to prune the obtained list (see Section~\ref{sec:Sudan_quantum}).
We omit the details, as they are not so relevant for us.
\end{remark}

As a consequence of Theorem~\ref{thm:nc0plus} and Theorem~\ref{thm:qcircuit}, we conclude that the problem of list decoding the Hadamard code separates the complexity classes \nc$^0[\oplus]$ and \qnc$^0[\oplus]$;
this holds for any positive error parameter $\delta>0$.
The task of proving quantum advantage for a natural problem such as list decoding was the original motivation for the present work.

In the high-error regime where the parameter $\delta$ approaches the information-theoretic limit of $1/2$ (which is relevant for hardness amplification), a stronger separation follows by combining Theorem~\ref{thm:qcircuit} with a result of Sudan showing hardness of noisy decoding by \ac$^0[\oplus]$ circuits (see Corollary~\ref{cor:Sudan_Had} below).\footnote{The same separation of complexity classes can also be obtained by combining other previously-known results; see Section~\ref{sec:literature} for a discussion.}
To state this separation theorem precisely, we consider the following problem:

\medskip
\noindent
\textbf{List-Hadamard problem:}
Let $\eps: \N\to (0,1]$ be a function.
For each dyadic number $n=2^k$ we define the problem $\LH_n(\eps)$ as follows:
given $y\in \F_2^n$, output a list of at most $n/4$ elements in $\F_2^k$ containing every $x\in \F_2^k$ satisfying $d\big(y, H(x)\big)\leq \big(\frac{1}{2} - \eps(n)\big) n$.
\medskip

The most general form of our quantum advantage result is given by the following theorem:

\begin{theorem}[Quantum-vs-classical separation] \label{thm:separation}
For any constant $\delta \in (0,\frac{1}{2})$, list decoding the Hadamard code with error parameter $\delta$ separates \qnc$^0[\oplus]$ from~\nc$^0[\oplus]$.
Moreover, for any $(\log n)/\sqrt{n} \leq \eps(n) \leq 1/(\log n)^{\omega(1)}$, the list-Hadamard problem $\LH_n(\eps)$ separates \qnc$^0[\oplus]$ from \ac$^0[\oplus]$.
\end{theorem}

\subsection{Related results and discussion}
\label{sec:literature}

Both the problem of decoding corrupted ECCs and the problem of proving quantum-versus-classical separations of complexity classes are well studied, and there are several results in the literature related to the results presented here.

\medskip
The main strength of our Theorem~\ref{thm:nc0plus} is that it holds for any code and for any positive error rate.
Complementary results are known for restricted classes of codes, and also for when the error rate tends to~$1/2$.
We will now expand on some of these results.

A code $E:\F_2^k\to\F_2^n$ is \emph{$t$-wise independent} if, for any $t$-subset of coordinates~$S\subseteq [n]$ and a uniformly random $X\in \F_2^k$, the restriction $E(X)_{|S}$ is uniformly distributed over~$\F_2^S$.
Many codes have this property;
for instance, the dual code of a linear code of distance~$d$ is $(d-1)$-wise independent.
Under the same noise model considered here, Lee and Viola~\cite{LeeViola:2017}, using earlier work of Viola~\cite{Viola:2009}, 
showed that~\nc$^0[\oplus]$ circuits cannot distinguish a corrupted uniformly random codeword of an $\omega(1)$-wise independent linear code from a uniformly random element of~$\F_2^n$.
Note that this problem is formally easier than (list) decoding.

Their result does not cover the Hadamard code, however, as it is not even 3-wise independent.
Indeed, the Hadamard code is also easy to distinguish, as it contains the sub-code $(x_1,x_2, x_1+x_2)$.
Since the parity of these three bits is always zero, the parity under noise is biased towards zero and therefore easily distinguished from the parity of a random string.

In the very-high-error regime where the error rate approaches the information-theoretic limit of $1/2$ (which is relevant for hardness amplification), stronger results are also known.
For instance, Sudan (see~\cite[Section~6.2]{viola:2006}) showed that list decoding with error parameter $1/2 - \eps$ requires probabilistic~\ac$^0[\oplus]$ circuits to have size $\exp(\poly(1/\eps))$.
Below we state his result when restricted to the Hadamard code, which is done for concreteness and better clarity;
as can be easily seen from its proof, one could instead consider any other ECC.

\begin{restatable}[MAJORITY from list-Hadamard]{theorem}{ThmSudanMajority}
\label{thm:sudan}
Let $\mathcal{C}$ be a probabilistic circuit that solves the list-Hadamard problem $\LH_n(\eps)$ with probability at least $3/4$.
There exists a (deterministic) oracle \ac$^0$ circuit $\mathcal{D}$ of size $\poly(n, 1/\eps)$ which, when given oracle access to $\mathcal{C}$ and the ability to fix its random bits, computes MAJORITY on $\Omega(1/\eps)$ bits.
\end{restatable}

This result can be readily deduced from Sudan's arguments exposed in~\cite[Section~6.2]{viola:2006};
since it is not given in this form elsewhere, we include its elegant proof in Appendix~\ref{sec:sudan}.
As a corollary, the circuit lower bound for MAJORITY due to Razborov~\cite{Razborov:1987} and Smolensky~\cite{Smolensky:1987} gives the following (known) hardness result for list decoding the Hadamard code.

\begin{corollary}[Hardness of list-Hadamard] \label{cor:Sudan_Had}
If $\eps(n) \leq 1/(\log n)^{\omega(1)}$, then the list-Hadamard problem $\LH_n(\eps)$ cannot be solved by a probabilistic \ac$^0[\oplus]$ circuit with probability~$\Omega(1)$.
\end{corollary}

Combining this corollary with our~\qnc$^0[\oplus]$ circuits for list-Hadamard given in Corollary~\ref{cor:qlisthad}, we obtain the second separation of complexity classes stated in Theorem~\ref{thm:separation}.

\medskip
The existence of the quantum circuits of Theorem~\ref{thm:qcircuit} and Corollary~\ref{cor:qlisthad} also follows from the Goldreich-Levin algorithm and the surprising fact that MAJORITY can be computed by a \qnc$^0[\oplus]$ circuit~\cite{HoyerSpalek:2005, TakahashiTani:2013}.\footnote{The above-mentioned classical hardness of MAJORITY thus also implies a separation between \ac$^0[\oplus]$ and \qnc$^0[\oplus]$, showing that despite its simplicity, the latter class of quantum circuits is remarkably powerful.}
However, whereas the circuit based on the Golreich-Levin algorithm depends on the error parameter~$\eps$ (which influences the size of the MAJORITY gates), our quantum circuit is constructed independently of~$\eps$.
Moreover, a key enabling sub-routine in the H{\o}yer-\v{S}palek circuit for MAJORITY~\cite{HoyerSpalek:2005} is the powerful quantum fan-out gate (see below for further details).
In our circuit for Corollary~\ref{cor:qlisthad}, we construct this gate explicitly using only classical parity gates and single- and two-qubit gates; 
these gates are native to many quantum architectures and as such, may give an easier way to implement quantum fan-out.
In the opposite direction, one can use the ideas behind the proof of Theorem~\ref{thm:sudan} to show that our quantum circuit from Corollary~\ref{cor:qlisthad} also gives a \qnc$^0[\oplus]$ circuit for MAJORITY, albeit not exact (see Section~\ref{sec:Sudan_quantum}).

Finally, quantum list-decoding of classical error correcting codes was also studied in~\cite{Yamakami:2016}.
The model considered there consists of a faulty quantum circuit that implements the encoding, which differs from our setting.

\subsection{Future directions}

We conjecture that the correct rate of decay in Theorem~\ref{thm:nc0plus} is exponential in the message length, as suggested by the results we obtain in the high-characteristic setting (exposed in Section~\ref{sec:high_char}).
This raises several intriguing questions related to notions of rank for tensors and polynomial maps~\cite{BrietCS:2022}.

Our results leave open the problem of decoding more general classes of error correcting codes by shallow quantum circuits, or by efficient quantum algorithms.
A particular class of interest consists of low-degree Reed-Muller codes, which generalize the Hadamard code.

A related question is if the problem of distinguishing random corrupted codewords of some code from uniformly random strings gives similar separations of the quantum and classical complexity classes considered here.
For example, Lee and Viola~\cite{LeeViola:2017} proved \nc$^0[\oplus]$-hardness of distinguishing $\omega(1)$-wise independent codes.
Is  there a \qnc$^0[\oplus]$ distinguisher for such a code?

\section{Techniques}
\label{sec:techniques}
To establish our main results, we use techniques from two different areas.
Broadly speaking, Theorem~\ref{thm:nc0plus} builds on ideas from higher-order Fourier analysis~\cite{Tao:2012,HHL}, while Theorem~\ref{thm:qcircuit} (unsurprisingly) uses ideas from quantum computing~\cite{nielsen_chuang_2010}.

\subsection{Polynomial equidistribution}
The proof of Theorem~\ref{thm:nc0plus} uses the basic observation that any function $\F_2^n\to\F_2^k$ that is computable by an \nc$^0[\oplus]$ circuit can be given by a collection of $k$ constant-degree polynomials over~$\F_2$ in~$n$ variables.
Indeed, any gate with fan-in~$d$ implements a function $\F_2^d\to\F_2$ and any such function can be represented by a $d$-variable polynomial of total degree at most~$d$.
Degree is multiplicative under composition and composition occurs only between different layers of the circuit.
Since the parities amount to addition in~$\F_2$ and \nc$^0$ circuits have constant depth, the total degree of the output is bounded.

We will therefore study the distribution of polynomial maps under biased input distributions.
We will do so in a slightly more general setting over arbitrary finite fields of prime order.\footnote{The restriction to prime order is done for notational reasons and for ease of exposition. Our arguments can be readily adapted to the case of non-prime finite fields.}
For a prime~$p$, let~$\F_p$ denote the finite field with~$p$ elements.
For $\rho\in [0,1]$, an $\F_p$-valued random variable~$Z$ is \emph{$\rho$-biased} if with probability~$\rho$ it equals~0 and with probability $1 - \rho$ it is uniformly distributed over~$\F_p$.
Note that this corresponds to the noise $\mathcal{N}_\rho(0)$ added by the symmetric channel when the alphabet is~$\F_p$.

A mapping $\phi:\F_p^n\to\F_p^k$ is a \emph{polynomial map} if there exist polynomials $f_1,\dots,f_k\in \F_p[x_1,\dots,x_n]$ such that $\phi = (f_1,\dots,f_k)$.
The degree of~$\phi$ is the maximal degree among the~$f_i$.
To prove Theorem~\ref{thm:nc0plus}, it thus suffices to prove the following result.

\begin{theorem}[Impossibility of decoding by polynomial maps]\label{thm:biased_equi}
For any $d\in \N$ and $\rho,\eps\in (0,1)$ there exists an integer $k_0 = k_0(p,d,\rho,\eps)$ such that the following holds.
Let $k \geq k_0$ and $n$ be integers, $\phi:\F_p^n\to\F_p^k$ be a polynomial map of degree at most~$d$ and $E:\F_p^k\to\F_p^n$ be an arbitrary function.
Then
\begin{equation*}
\Pr_{x\in \F_p^k, Z\sim \mathcal N_\rho(0)}\big[\phi\big(E(x) + Z\big) = x\big] \leq \eps.
\end{equation*}
\end{theorem}

Studying the distribution of polynomial maps in many variables over a finite field falls within the purview of additive combinatorics.
In the ``unbiased'' situation where~$Z$ is uniformly distributed there are powerful tools from higher-order Fourier analysis that can be used to study the distribution of~$\phi(Z)$.
In particular, Green and Tao~\cite{GreenT:2009} proved that if~$\phi$ is ``regular'' (random-like), then~$\phi(Z)$ is approximately uniformly distributed over~$\F_p^k$.
This implies that the probability of the event $\{\phi(E(x) + Z) = x\}$ considered is small for every~$x$.
A ``regularity-type'' lemma proved in~\cite{GreenT:2009} shows that one can ``force'' $\phi$ to be regular by restricting it to a partition defined by sufficiently many polynomial equations of degree less than the degree of~$\phi$.
However, these techniques cause the size of the polynomial map $\phi$ considered to blow up considerably, and are only effective if~$k$ is an extremely slowly growing function of~$n$.

In order to deal with this issue, and to adapt these results to the case where~$Z$ is no longer uniform but biased, we employ a dichotomy often used in additive combinatorics that studies the ``pseudorandom'' case of regular maps separately from the ``structured'' case of maps that carry a certain algebraic structure.
This is done by defining and studying a new notion of rank for (high-dimensional) polynomial maps, which we call the \emph{analytic rank},\footnote{A very similar notion of rank was defined for multilinear forms by Gowers and Wolf~\cite{GowersW:2011}, who coined the term analytic rank. We use the same name to highlight the similarity between our two notions, which are relevant for distinct types of mathematical objects.}
and which measures how equidistributed the values taken by the considered map are.

In the pseudorandom case, a key tool we use is a new random restriction result for high-rank polynomial maps proved in a companion paper~\cite{BrietCS:2022}.
We use this to show that the distribution of values taken by a high-rank polynomial map will be close to uniform even under a biased input distribution.
This implies that the event considered in the theorem has very low probability for any fixed~$x$, in which case we can conclude by averaging.

In the structured case we deal instead with polynomial maps of low rank, whose values are in a sense poorly distributed.
Results from higher-order Fourier analysis then imply that they can be determined by ``few'' lower-degree polynomial maps (plus a few extra polynomials);
by a simple Fourier-analytic argument we can reduce the analysis of a low-rank polynomial map to those lower-degree maps which specify it, making it amenable to an inductive argument.

\subsection{Building the quantum circuit}
The quantum circuit of Theorem~\ref{thm:qcircuit} is inspired by a distributed version of the Bernstein-Vazirani algorithm~\cite{BV:1997}.
Given a corrupted Hadamard codeword~$H(x)$, this single-query quantum algorithm returns~$x$ with probability~$\Omega(\eps^2)$.
The distributed version describes an entangled strategy for a particular non-local game~\cite{CHTW:2004} consisting of~$n$ players who, when given unique coordinates of~$H(x)$, must each return an element of~$\F_2^k$.
They win if and only if the sum of their answers equals~$x$.
It turns out that by sharing an $n$-partite GHZ state of local dimension~$2^k$, they can simulate the Bernstein-Vazirani algorithm and achieve the same success probability.
We then turn this entangled strategy into a quantum circuit that only uses single and two-qubit gates and classical parity gates. 
For this we use two constant-depth sub-routines, one for preparing GHZ states and another for the quantum \emph{fan-out gate}~\cite{PhamSvore:2013}, which implements the map $\ket{x}\ket{y_1}\hdots\ket{y_n}\mapsto\ket{x}\ket{y_1\oplus x}\hdots\ket{y_n\oplus x}$.

To generate the GHZ state, we use a poor man's cat state~\cite{WattsKothariSchaefferTal:2019}, which is a GHZ state with some of its qubits flipped. We correct this poor man's cat state to a GHZ state by flipping qubits based on parity computations. The input for these parity computations follows from the procedure that generates the poor man's cat state. 

To implement the quantum fan-out gate, we use ideas from distributed quantum computing. These ideas use GHZ states and classical parity gates together with single and two-qubit gates. With the quantum fan-out gate, we also obtain the quantum parity  gate, by conjugating the quantum fan-out gate with Hadamard gates. 

Part of the circuit is applying phase-flips, conditional on the bits of the corrupted codeword. To do this, we use quantum fan-out gates in a circuit, exponential in size in $k$ to correctly apply the phase-flips~\cite{TakahashiTani:2013}. 

The depth of the list-decoding circuit is constant, whereas the circuit size is $O(n^2\log n)$. We also show how to reduce this complexity to $O(n\log n\log\log n)$, while increasing the depth by only a small constant number. We do this by preparing a state on $\lceil\log(k+1)\rceil$ qubits. Evaluating an OR on this newly prepared state yields the same result as evaluating an OR on the original $k$ qubits~\cite{HoyerSpalek:2005}. Applying the same exponential size circuit as before on this newly prepared state indeed gives the reduced circuit size.

\section{Warm-up: The linear case and a non-local game}

This section is meant to give some intuition for the proofs of Theorem~\ref{thm:biased_equi} and Theorem~\ref{thm:qcircuit}, as well as provide the first steps in those proofs.

\subsection{Impossibility of decoding for linear maps} \label{sec:cstrategy}

To motivate our later arguments, here we present a proof of the first nontrivial case of Theorem \ref{thm:biased_equi}, namely that of maps~$\phi: \F_p^n \to \F_p^k$ of degree~$1$.
In this case, there is a matrix~$U\in \F_p^{k\times n}$ and a vector~$v\in \F_p^k$ such that
\beqn
\phi(y) = Uy + v \quad \text{for all } y\in \F_p^n.
\eeqn

Let~$x$ be a uniformly distributed random variable over~$\F_p^k$ and~$Z$ be an~$\mathcal N_\rho(0)$-distributed random variable over~$\F_p^n$.
Our goal is then to bound the probability of the event
\beq\label{eq:baseeq}
U(E(x) + Z) + v = x.
\eeq

We distinguish two cases based on the rank of~$U$.
Let~$r \in [k]$ be an integer to be set later.
If~$U$ has rank at most~$r$, then its image~$\im(U)$ is a subspace of size at most~$p^r$.
If \eqref{eq:baseeq} holds, then~$x$ is contained in the coset~$v + \im(U)$ of this subspace, which (for~$x$ uniform over~$\F_p^k$) happens with probability at most~$p^r/p^k$.
Hence, \eqref{eq:baseeq} holds with probability at most~$p^{-(k-r)}$ in this case.

For the ``pseudorandom case'' of high-rank matrices, we make the following simple but important observation:
one can sample~$Z \sim \mathcal N_\rho(0)$ by first sampling the set~$I \subseteq [n]$ of ``corrupted coordinates'', then sampling the ``noise'' $y$ uniformly at random from~$\F_p^I$ and setting\footnote{Given $x\in \F_p^n$ and $I\subseteq [n]$, we denote by $x_{|I}\in \F^I$ the restriction of $x$ to the coordinates indexed by $I$.}
$Z_{|I} = y$, $Z_{|[n]\setminus I} = 0$.
Each index~$i\in [n]$ has probability~$1-\rho$ of belonging to the random set~$I$, with these events being mutually independent;
we denote this sampling scheme by~$I\sim [n]_{1-\rho}$.

If we denote by~$U_I\in \F_p^{k\times I}$ the restriction of~$U$ to the columns indexed by~$I\subseteq [n]$, it follows that the random variable~$UZ$ has the same distribution as the random variable~$U_I y$, where $I\sim [n]_{1-\rho}$ and $y$ is uniformly distributed over $\F_p^I$.
Thus, for any given~$x\in \F_p^k$, we have
\beqn
\Pr_{Z\sim \mathcal{N}_\rho(0)}\big[ U(E(x)+Z) + v = x \big]
=
\Exp_{I\sim [n]_{1-\rho}} \Pr_{y\in F^I}\big[ U_I y = x - UE(x) - v \big].
\eeqn
Now, if~$I \subseteq [n]$ is fixed and~$y$ is uniformly distributed over~$\F_p^I$, then the random variable~$U_I y$ is uniformly distributed over~$\im(U_I)$;
hence
\beqn
\max_{w\in \F_p^k} \Pr_{y\in F^I}\big[ U_I y = w \big] =  \frac{1}{|\im(U_I)|} =  \frac{1}{p^{\rank(U_I)}}.
\eeqn
Taking the expectation over~$I\sim [n]_{1-\rho}$ and~$x\in \F_p^k$, we conclude that event \eqref{eq:baseeq} holds with probability at most~$\Exp_{I\sim [n]_{1-\rho}} p^{-\rank(U_I)}$.

Suppose now that~$U$ has rank at least~$r$, and let~$J\subseteq [n]$ be a set of~$r$ linearly independent columns of~$U$.
By the Chernoff bound (see \eg~\cite{HagerupRub:1990}), we have that
\beqn
\Pr_{I\sim[n]_{1-\rho}}\bigg[|I\cap J| \leq  \frac{(1-\rho) r}{2}\bigg] \leq e^{-(1-\rho) r/8}.
\eeqn
Thus~$U_I$ will contain more than~$(1-\rho)r/2$ linearly independent columns with probability at least~$1 - e^{-(1-\rho) r/8}$;
whenever this happens we have~$\rank(U_I) \geq (1-\rho)r/2$.
It follows that
\begin{align*}
    \Pr_{x\in \F_p^k, Z\sim \mathcal N_\rho(0)}\big[ U(E(x) + Z) + v = x \big]
    &\leq \Exp_{I\sim [n]_{1-\rho}} p^{-\rank(U_I)} \\
    &\leq e^{-(1-\rho) r/8} + p^{-(1-\rho) r/2}
\end{align*}
in this ``high-rank'' case.

Setting~$r = k/2$ (say) implies that in both cases the probability that event \eqref{eq:baseeq} holds decays exponentially in~$k$, which concludes the analysis.

\subsection{Quantum decoding in a non-local game}
\label{sec:qstrategy}
Our quantum algorithm is inspired by the analysis of a particular non-local game.
In a non-local game, a referee randomly sends questions to a set of players, according to a probability distribution known to the players in advance.
Then, without communicating with each other, the players individually answer the referee.
Finally, the referee determines whether the players win or lose based solely on the questions and answers.
The rule used by the referee is known to the players in advance as well.
With a (deterministic) classical strategy, the players decide before the game starts what to answer to each possible question.
With an entangled strategy, the players base their answers on the outcomes of local measurements of their respective parts of a shared entangled state.
We refer to~\cite{CHTW:2004} for further background on non-local games.

Let $H:\F_2^k\to \F_2^n$ be the Hadamard code, where $n = 2^k$ and let $\eps \in (0,1/2)$ be a constant (independent of~$n$).
We identify the codewords~$H(x)$ with functions $\F_2^k\to\F_2$ given by~$H(x)(y) = \langle x,y\rangle$.
We consider the following non-local game, which we shall refer to as the \emph{Hadamard game}.
There are~$n$ players, each labeled uniquely with an element in~$\F_2^k$.
The referee picks a uniformly chosen message~$x\in \F_2^k$ and randomly corrupts the codeword~$H(x)$ using the binary symmetric channel with error rate $1/2 - \eps$, resulting in a function~$c:\F_2^k\to \F_2$.
He then sends player~$y$ the value~$c(y)$.
The players each return a string in~$\F_2^k$ and they win the game if the sum of their answers equals~$x$.

Here we show that entangled players can win the Hadamard game with probability~$\Omega(1)$.
The corresponding strategy is inspired by the famous Bernstein-Vazirani algorithm~\cite{BV:1997}. 
The strategy is based on an $n$-partite GHZ state of local dimension~$2^k$, shared by the~$n$ players:
\beqn
\frac{1}{\sqrt{n}}\sum_{y\in \F_2^k} \ket{y}\otimes\hdots\otimes\ket{y}.
\eeqn

Upon receiving their input~$c(y)$, player~$y$ applies a conditional phase flip on their part of the shared state:
\beq\label{eq:condphase}
\ket{z}\mapsto\begin{cases}(-1)^{c(y)}\ket{z} & \text{if } z=y \\  \ket{z} & \text{else}\end{cases}.
\eeq
Once all players have done this, they share the state
\begin{equation*}
\frac{1}{\sqrt{n}}\sum_y (-1)^{c(y)}\ket{y}\otimes\hdots\otimes\ket{y}. 
\end{equation*}

Each player then applies a $k$-qubit Hadamard gate to their local register and measures in the computational basis. 
The measurement results are then returned by each player. 
The state before the measurements is:
\begin{equation*}
n^{-(n+1)/2}\sum_{y\in \F_2^k} \sum_{b_1,\hdots,b_n\in\F^k_2}(-1)^{c(y)}(-1)^{\langle y,b_1+\cdots +b_n\rangle}\ket{b_1}\otimes\hdots\otimes\ket{b_n}. 
\end{equation*}

The probability that the measurement results sum to a string $z\in\F^k_2$ is therefore given by
\begin{align*}
\Pr\bigg[\sum_{i=1}^nb_i= z\bigg] & = \frac{1}{n^{(n+1)}}\sum_{b_1+\cdots+b_n=z}\bigg| \sum_{y\in\F_2^k} (-1)^{c(y)+\langle y,z\rangle}\bigg|^2 \\
& = \bigg| 1 - 2\frac{d\big(c, H(z)\big)}{n}\bigg|^2.
\end{align*}
It follows from the Chernoff bound~\cite{HagerupRub:1990} that for fixed~$x\in \F_2^k$ and the random~$c$ obtained by corrupting the codeword~$H(x)$,
\beqn
\Pr\bigg[\frac{d\big(c, H(x)\big)}{n} \geq \frac{1-\eps}{2}\bigg] \leq \exp(-C\eps^2 n).
\eeqn
Hence, by the union bound, for fixed $\eps \in (0,1/2)$, the players win with probability at least $C\eps^2$, where the probability is taken over the message~$x$, the noise corrupting the codeword~$H(x)$ to~$c$ and the measurements done by the players.

Note that this strategy in fact succeeds with probability~$C\eps^2$ for \emph{every}~$x$ and whenever at most any $(1/2 - \eps)$-fraction of the coordinates of~$H(x)$ are flipped.

\section{Classical hardness of list decoding}
\label{sec:classical}

In this section we will prove Theorem~\ref{thm:biased_equi}, which -- as explained before -- implies that NC$^0[\oplus]$ circuits are unable to perform list decoding, no matter which specific code is considered
(see Theorem~\ref{thm:nc0plus} for a formal statement).
We will do so by following a similar strategy as we used for maps of degree~$1$ in Section~\ref{sec:cstrategy}, dividing the analysis into the ``pseudorandom'' case of high-rank polynomial maps and the ``structured'' case of low-rank polynomial maps.

\medskip

There is a well-studied notion of rank for polynomials~$P\in \F_p[x_1,\dots,x_n]$, first introduced by Green and Tao~\cite{GreenT:2009}, which is defined (roughly speaking) as the smallest number of lower-degree polynomials needed to compute~$P$.
A different notion of rank, called the analytic rank, was later introduced by Gowers and Wolf~\cite{GowersW:2011} when studying linear systems of equations over~$\F_p^n$.
It is related to the \emph{bias} of the polynomial~$P$, or more specifically to the bias of the symmetric~$\deg(P)$-multilinear form associated to~$P$.
The bias of a function $f: \F_p^n\to \F_p$ is an analytic measure of how well-equidistributed the values of~$f$ are when evaluated on a uniformly random input;
formally,
\beq\label{def:bias}
\bias(f) = |\Exp_{x\in \F_p^n} \omega^{f(x)}|
\eeq
where we write~$\omega = e^{2i\pi /p}$ for a primitive~$p$-th root of unity.

When dealing with a polynomial~$P$ of some bounded degree~$d$, having non-negligible bias implies that it has a significant amount of internal structure.
Such a result was first proven by Green and Tao~\cite{GreenT:2009} in the case of polynomials whose degree~$d$ is smaller than the characteristic~$p$ of the field considered, and motivated the introduction of both their notion of rank and Gowers and Wolf's notion of analytic rank.
We will need a similar result, proven by Kaufman and Lovett~\cite{KaufmanLovett:2008}, which generalizes this theorem to characteristics~$p \leq d$ and also gives more precise information on the structure of the polynomial.

For a vector~$h\in \F_p^n$ and a polynomial~$P\in \F_p[x_1,\dots,x_n]$, the derivative of~$P$ in direction~$h$ is defined by
\beqn
\Delta_h P(x) = P(x+h) - P(x).
\eeqn
Note that~$\Delta_h P$ is also a polynomial on~$\F_p^n$, and (as with usual derivatives in real analysis) its degree is strictly smaller than~$\deg(P)$.
The derivatives of a polynomial map~$\phi:\F_p^n\to \F_p^k$ are defined analogously, and also satisfy~$\deg(\Delta_h \phi) < \deg(\phi)$.

The following result of Kaufman and Lovett shows that polynomials with large bias must be highly structured:

\begin{theorem}[Bias implies low rank]\label{thm:KL}
For every~$d\in \N$ and~$\eps>0$, there is an~$r = r(p,d,\eps)\in \N$ such that the following holds.
If~$P\in \F_p[x_1,\dots,x_n]$ is a polynomial of degree at most~$d$ with~$\bias(P) \geq \eps$, then there exist~$h_1,\dots,h_r\in \F_p^n$ and a map~$\Gamma:\F_p^r\to\F_p$ such that
\beqn
P(x) \equiv \Gamma\big(\Delta_{h_1}P(x),\dots, \Delta_{h_r}P(x)\big).
\eeqn
\end{theorem}

\subsection{The analytic rank of polynomial maps}

Inspired by Gowers and Wolf's notion of analytic rank for multilinear forms and polynomials~\cite{GowersW:2011}, we introduce a new notion of rank for higher-dimensional polynomial maps which we also call the \emph{analytic rank}.
This notion will be crucial in our proof of Theorem~\ref{thm:biased_equi};
intuitively, it measures how well a given polynomial map~$\phi$ can be approximated by lower-degree maps.

For integers $d, n, k \geq 1$, we denote by $\Pol_{\leq d}(\F_p^n, \F_p^k)$ the space of all polynomial maps $\phi: \F_p^n \to \F_p^k$ of degree at most~$d$.

\begin{definition}[Analytic rank]\label{def:phirank}
Given a polynomial map $\phi\in \Pol_{\leq d}(\F_p^n, \F_p^k)$, we define its analytic rank $\arank_d(\phi)$ by
\beqn
\arank_d(\phi) = -\log_p \bigg( \max_{\psi:\F_p^n\to\F_p^k,\, \deg(\psi)<d} \Pr_{x\in \F_p^n} \big[\phi(x) = \psi(x)\big] \bigg).
\eeqn
\end{definition}

Note that, for affine-linear maps $\phi \in \Pol_{\leq 1}(\F_p^n, \F_p^k)$, this definition coincides with the usual notion of rank for the matrix~$U \in \F_p^{k\times n}$ encoding its linear part.
Indeed, write~$\phi(x) = Ux + v$ for some~$v \in \F_p^k$.
Since~$Ux$ is uniformly distributed over~$\im(U) \simeq \F_p^{\rank(U)}$ when~$x$ is uniformly distributed over~$\F_p^k$, we have that
\beqn
\Pr_{x\in \F_p^n} \big[Ux+v = w\big] =
\begin{cases}
 p^{-\rank(U)} &\text{if } w-v \in \im(U), \\
 0 &\text{if } w-v \notin \im(U).
\end{cases}
\eeqn
This might help explain the reason for the~$-\log_p$ in the definition of analytic rank, as well as the need to maximize the probability of agreement over all lower-degree maps.

Another useful way of viewing the analytic rank of a polynomial map~$\phi$ is as a measure of how well-equidistributed its values are in~$\F_p^k$, up to lower-degree perturbations.
Indeed, we can equivalently write
\beqn
\arank_d(\phi) = \min_{\psi:\F_p^n\to\F_p^k,\, \deg(\psi)<d}-\log_p \big(\Exp_{v\in \F_p^k, x\in \F_p^n}\omega^{\langle v,\, \phi(x) - \psi(x)\rangle}\big).
\eeqn
The expectation inside the logarithm above is analogous to the notion of bias \eqref{def:bias} given before, and can be seen as an analytic measure of how close to uniformly distributed over~$\F_p^k$ the values taken by~$\phi-\psi$ are.

It is clear from the definition that the function $\arank_d$ is non-negative (since probabilities are bounded by~$1$), and that $\arank_d(\phi) = 0$ if and only if $\deg(\phi) \leq d-1$.
It also satisfies several useful properties in common with the rank of matrices;
in order to state them we will need some notation for considering coordinate restrictions:

\begin{definition}[Restriction]
For a polynomial map~$\phi:\F_p^n\to\F_p^k$ and subset~$I\subseteq [n]$, we define the restriction~$\phi_{|I}:\F_p^I\to\F_p^k$ to be the map given by
$\phi_{|I}(y) = \phi(\bar y)$,
where~$\bar y\in \F_p^n$ agrees with~$y$ on the coordinates in~$I$ and is zero elsewhere.
\end{definition}

The properties of analytic rank which will be important to us are summarized in the next lemma.
Those rank functions for polynomial maps which satisfy all these properties are called \emph{natural rank functions} in~\cite{BrietCS:2022}.

\begin{lemma}[Properties of analytic rank] \label{lem:natural}
For all integers $d, n, k \geq 1$, the analytic rank function $\arank_d$ satisfies:
\begin{enumerate}
    \item Symmetry: \\
    $\arank_d(\phi) = \arank_d(-\phi)$ for all~$\phi\in \Pol_{\leq d}(\F^n, \F^k)$.
    \item Sub-additivity: \\
    $\arank_d(\phi + \gamma) \leq \arank_d(\phi) + \arank_d(\gamma)$ for all~$\phi, \gamma\in \Pol_{\leq d}(\F^n, \F^k)$.
    \item Monotonicity under restrictions: \\
    $\arank_d(\phi_{|I}) \leq \arank_d(\phi)$ for all~$\phi\in \Pol_{\leq d}(\F^n, \F^k)$ and all sets~$I \subseteq [n]$.
    \item Restriction Lipschitz property: \\
    $\arank_d(\phi_{|I \cup J}) \leq \arank_d(\phi_{|I}) + |J|$ for all~$\phi\in \Pol_{\leq d}(\F^n, \F^k)$ and all sets~$I, J \subseteq [n]$.
\end{enumerate}
\end{lemma}

\begin{proof}
The first property is trivial.
To prove property (2), let~$\psi, \chi: \F_p^n\to \F_p^k$ be polynomial maps of degree at most~$d-1$ such that
\begin{align*}
    &\arank_d(\phi) = -\log_p \Pr_{x\in \F_p^n} \big[\phi(x) = \psi(x)\big], \\
    &\arank_d(\gamma) = -\log_p \Pr_{x\in \F_p^n} \big[\gamma(x) = \chi(x)\big].
\end{align*}
Then~$p^{-\arank_d(\phi) -\arank_d(\gamma)}$ can be expressed as
\begin{align*}
    \Pr_{x, y\in \F_p^n} &\big[\phi(x) = \psi(x) \And \gamma(y) = \chi(y)\big] \\
    &= \Pr_{x, y} \big[\phi(x) = \psi(x) \And \phi(x) + \gamma(x+y) = \psi(x) + \chi(x+y)\big],
\end{align*}
where we performed the change of variables~$(x, y)\mapsto (x, x+y)$.
Since
\beqn
\gamma(x+y) = \gamma(x) + \Delta_y \gamma(x),
\eeqn
this equals
\begin{align*}
\Pr_{x, y} \big[\phi(x) = \psi(x) &\And \phi(x) + \gamma(x) = \psi(x) + \chi(x+y) - \Delta_y \gamma(x)\big] \\
&\leq \Pr_{x, y} \big[\phi(x) + \gamma(x) = \psi(x) + \chi(x+y) - \Delta_y \gamma(x)\big].
\end{align*}

Note that, for any fixed~$y\in \F_p^n$, the function
\beqn
x \mapsto \psi(x) + \chi(x+y) - \Delta_y \gamma(x)
\eeqn
is a polynomial map of degree at most~$d-1$.
The last probability above is then bounded by
\begin{align*}
\max_y \Pr_x \big[\phi(x) + \gamma(x) = \psi(x) &+ \chi(x+y) - \Delta_y \gamma(x)\big] \\
&\leq \max_{\zeta: \F_p^n \to \F_p^k,\, \deg(\zeta) < d} \Pr_x \big[\phi(x) + \gamma(x) = \zeta(x)\big] \\
&= p^{-\arank_d(\phi + \gamma)}.
\end{align*}
Sub-additivity now follows by taking logarithms.

To prove property (3) it suffices to show that $\arank_d(\phi_{|[n]\setminus\{i\}})\leq \arank_d(\phi)$ for any $i\in[n]$, which can then be applied iteratively.
Assume for notational convenience that~$i=n$, and let $\psi:\F_p^n\to\F_p^k$ be a polynomial map of degree at most~$d-1$ which satisfies
\beqn
p^{-\arank_d(\phi)} = \Pr_{x\in \F_p^n}[\phi(x) = \psi(x)].
\eeqn
Factoring out the variable~$x_n$ allows us to write the probability on the right-hand side as
\begin{align*}
\Exp_{x_n\in \F_p} \Pr_{y\in \F_p^{n-1}}\big[\phi_{|[n-1]}(y) + \phi'(y,x_n)x_n = \psi_{|[n-1]}(y) + \psi'(y,x_n)x_n\big],
\end{align*}
where~$\phi'$ and~$\psi'$ are some polynomial maps of degree at most~$d-1$.
By the averaging principle, this is at most
\begin{align*}
\max_{x_n\in \F_p}\Pr_{y\in \F_p^{n-1}} &\big[\phi_{|[n-1]}(y) = \psi_{|[n-1]}(y) + \psi'(y,x_n)x_n - \phi'(y,x_n)x_n \big] \\
&\leq \max_{\zeta: \F_p^{n-1} \to \F_p^k,\, \deg(\zeta) < d} \Pr_{y\in \F_p^{n-1}} \big[\phi_{|[n-1]}(y) = \psi_{|[n-1]}(y) + \zeta(y)\big] \\
&= p^{-\arank_d(\phi_{|[n-1]})},
\end{align*}
showing that $\arank_d(\phi_{|[n-1]})\leq \arank_d(\phi)$ as wished.

Finally, for the Lipschitz property (4), let~$\psi: \F_p^I \to \F_p^k$ be a map with $\deg(\psi)<d$ maximizing the agreement probability $\Pr_{x\in \F_p^I} \big[\phi_I(x) = \psi(x)\big]$, and suppose without loss of generality that $J\cap I = \emptyset$.
Then
\begin{align*}
    p^{-\arank_d(\phi_{|I \cup J})}
    &\geq \Pr_{x\in \F_p^I,\, y\in \F_p^J} \big[\phi_{|I \cup J}(x,y) = \psi(x)\big] \\
    &\geq \Pr_{x\in \F_p^I,\, y\in \F_p^J} \big[\phi_{|I \cup J}(x,0) = \psi(x) \And y=0\big] \\
    &= p^{-|J|}\, \Pr_{x\in \F_p^I} \big[\phi_I(x) = \psi(x)\big] \\
    &= p^{-\arank_d(\phi_I) - |J|},
\end{align*}
and the restriction Lipschitz property follows.
\end{proof}

\subsection{Biased equidistribution of high-rank maps}

As in the degree-$1$ case considered in Section~\ref{sec:cstrategy}, we will need to
study the distribution of values~$\phi(Z)$ taken by a polynomial map~$\phi$ when the input is a $\rho$-biased random variable~$Z\sim \mathcal{N}_{\rho}(y)$.
This can be done by considering restrictions of~$\phi$ to random subsets of variables, which model the coordinates ``corrupted'' by the studied random process.

Motivated by this problem, the behavior of rank functions under random coordinate restrictions was studied in detail by the first and third authors~\cite{BrietCS:2022}.
In the nomenclature of that paper, Lemma~\ref{lem:natural} shows that the analytic rank is a natural rank function.
Applying \cite[Theorem~1.8]{BrietCS:2022} we then immediately obtain the following result, which shows that random restrictions of a high-rank polynomial map will also have high rank with high probability.
(Recall that~$I\sim [n]_{\sigma}$ denotes the random process of sampling a subset~$I\subseteq [n]$ where each~$i\in [n]$ belongs to~$I$ with probability~$\sigma$, all events being mutually independent.)

\begin{theorem}[Random restriction theorem]\label{thm:restriction}
For every~$d\in \N$ and~$\sigma,\eps\in (0,1]$, there exist~$\kappa = \kappa(d, \sigma)>0$ and~$R = R(d,\sigma,\eps)\in\N$ such that the following holds.
For every map~$\phi\in \Pol_{\leq d}(\F^n, \F^k)$ with~$\arank_d(\phi) \geq R$, we have that
\beqn
\Pr_{I\sim [n]_{\sigma}}\big[ \arank_d(\phi_{|I}) \geq \kappa\cdot \arank_d(\phi)\big] \geq 1 - \eps.
\eeqn
\end{theorem}

With the help of this theorem, it is easy to show that high-rank polynomial maps are approximately equidistributed even under biased inputs:

\begin{lemma}[Biased equidistribution lemma]\label{lem:equidistribution}
For every~$d\in \N$ and $\rho, \eps \in (0, 1)$ there exists a constant $R_0 = R_0(d, \rho, \eps)>0$ such that the following holds.
If $\phi\in \Pol_{\leq d}(\F^n, \F^k)$ satisfies~$\arank_d(\phi) \geq R_0$, then
\beqn
\Pr_{Z \sim \mathcal N_\rho(0)}\big[ \phi(y+Z) = w \big] \leq \eps \quad \text{for all $y\in \F_p^n, w\in \F_p^k$.}
\eeqn
\end{lemma}

\begin{proof}
It suffices to prove the special case where both $y$ and $w$ are zero, that is
\beqn
\Pr_{Z \sim \mathcal N_\rho(0)}\big[ \phi(Z) = 0 \big] \leq \eps.
\eeqn
Indeed, for fixed~$y\in \F_p^n$ and~$w\in \F_p^k$, the map~$\tilde{\phi}: x \mapsto \phi(y+x)-w$ has the same degree and same analytic rank as~$\phi$, and satisfies~$\tilde{\phi}(x) = 0$ if and only if~$\phi(y+x)=w$.

We can sample~$Z \sim \mathcal N_\rho(0)$ by first sampling~$I \sim [n]_{1-\rho}$ (the ``corrupted coordinates''), then sampling~$z$ uniformly from~$\F_p^I$ (the ``noise'') and setting~$Z_{|I} = z$, $Z_{|[n]\setminus I} = 0$;
thus
\begin{align*}
\Pr_{Z \sim \mathcal N_\rho(0)}\big[\phi(Z) = 0\big]
&= \Exp_{I\sim [n]_{1-\rho}} \Pr_{z\in \F_p^I} \big[\phi_{|I}(z) = 0 \big] \\
&\leq \Exp_{I\sim [n]_{1-\rho}} p^{-\arank_d(\phi_{|I})}.
\end{align*}
Let~$R = R(d, 1-\rho, \eps/2)$ and~$\kappa = \kappa(d, 1-\rho)$ be the constants guaranteed by Theorem~\ref{thm:restriction}.
If~$\arank_d(\phi) \geq R$, from that result we obtain
\beqn
\Exp_{I\sim [n]_{1-\rho}} p^{-\arank_d(\phi_{|I})} \leq \eps/2 + p^{-\kappa\cdot \arank_d(\phi)}.
\eeqn
Taking\footnote{Note that this bound is non-increasing on the value of~$p$, so we can obtain a field-independent bound by considering the smallest case~$p=2$.}
$R_0 = \max\big\{ R,\, \log_p(2/\eps)/\kappa \big\}$ we conclude that
\beqn
\Pr_{Z \sim \mathcal N_\rho(0)}\big[\phi(Z) = 0\big] \leq \Exp_{I\sim [n]_{1-\rho}} p^{-\arank_d(\phi_{|I})} \leq \varepsilon
\eeqn
whenever~$\arank_d(\phi) \geq R_0$, as wished.
\end{proof}

\subsection{The proof of Theorem~\ref{thm:biased_equi}}

We are now ready to present the proof of Theorem~\ref{thm:biased_equi}, which proceeds by induction on the degree~$d$.
For degree-$1$ maps the result was already proven in the warm-up section,\footnote{It would also be possible to start the induction from the trivial base case~$d=0$ of constant maps, but we thought it more instructive to first present the argument for degree-$1$ maps in order to gain some intuition.}
so let~$d \geq 2$ and assume the result holds for maps of degree at most~$d-1$.

As was done in the base case, we will divide the argument into two parts, corresponding to whether the analytic rank of~$\phi$ is ``high'' (the pseudorandom case) or ``low'' (the structured case).
The pseudorandom case immediately follows from Lemma~\ref{lem:equidistribution}, the biased equidistribution lemma:
let~$R_0 = R_0(d, \rho, \eps)$ be the constant guaranteed by that lemma, and suppose that~$\arank_d(\phi) > R_0$.
Then for every~$x \in \F_p^k$ we have that
\beqn
\Pr_{Z \sim \mathcal N_\rho(0)}\big[ \phi \big(E(x)+Z\big) = x \big] \leq \eps,
\eeqn
and we conclude by averaging over all such~$x$.

Now suppose that~$\arank_d(\phi) \leq R_0$, and let~$\psi: \F_p^n \to \F_p^k$ be a map of degree at most~$d-1$ such that~$\Pr_{x\in \F_p^n} \big[\phi(x) = \psi(x)\big] \geq p^{-R_0}$.
Denote~$\tilde{\phi} = \phi-\psi$ for convenience, and let~$P\in \F_p[y_1,\dots,y_n,v_1,\dots,v_k]$ be the polynomial given by 
\beqn
P(y,v) = \langle v, \tilde{\phi}(y) \rangle.
\eeqn
This polynomial has non-negligible bias:
\beqn
\bias(P) = \Exp_{y\in \F_p^n} \Exp_{v\in \F_p^k} \omega^{\langle v, \tilde{\phi}(y) \rangle} = \Exp_{y\in \F_p^n} \mathbf{1} \big[\tilde{\phi}(y)=0\big] \geq p^{-R_0},
\eeqn
where $\omega = e^{2\pi i/p}$.
By Theorem~\ref{thm:KL}, there exist~$s = s(p, d, R_0)\in \N$, pairs $(h_1,w_1), \dots, (h_s,w_s)\in \F_p^n\times \F_p^k$ and a map~$\Gamma:\F_p^s\to\F_p$ such that
\beqn
P(y,v) = \Gamma\big(\Delta_{(h_1,w_1)}P(y,v),\dots, \Delta_{(h_s,w_s)}P(y,v)\big).
\eeqn

Let~$f:\F_p^s\to \C$ be the map given by~$f(t) = \omega^{\Gamma(t)}$ and let~$\widehat f:\F_p^s\to \C$ be its Fourier transform,
\beqn
\widehat f(\alpha) = \Exp_{t\in \F_p^s}f(t)\omega^{-\langle \alpha,t\rangle}.
\eeqn
Since~$P$ is linear in the last~$k$ coordinates, it follows that
\begin{align*}
\Delta_{(h,w)}P(y,v) 
&=
P(y+h, v+w) - P(y+h,v) + P(y+h,v) - P(y,v)\\
&=
\langle w, \tilde{\phi}(y+h)\rangle + \langle v, \Delta_h\tilde{\phi}(y)\rangle.
\end{align*}
By the Fourier inversion formula, we conclude that
\begin{align*}
\omega^{P(y,v)}
&=
f\big(\Delta_{(h_1,w_1)}P(y,v),\dots, \Delta_{(h_s,w_s)}P(y,v)\big)\\
&=
\sum_{\alpha\in \F_p^s}\widehat f(\alpha) \omega^{Q_\alpha(y) + \langle v, \gamma_\alpha(y)\rangle},
\end{align*}
where for~$\alpha\in \F_p^s$ we denote
\begin{align*}
Q_\alpha(y) &= \sum_{i=1}^s \langle \alpha_i w_i, \tilde{\phi}(y+h_i)\rangle,\\
\gamma_\alpha(y) &= \sum_{i=1}^s\alpha_i \Delta_{h_i}\tilde{\phi}(y).
\end{align*}
Note crucially that~$\deg(\gamma_\alpha) \leq d-1$ for all~$\alpha \in \F_p^s$, which is what will eventually allow us to apply the induction hypothesis.

It follows from our expression for~$\omega^{P(y, v)}$ that
\begin{align*}
\mathbf{1}[\phi(y) = x]
&=
\Exp_{v\in \F_p^k}\omega^{\langle v, \phi(y) - x\rangle}\\
&=
\Exp_{v\in \F_p^k}\omega^{P(y,v) +\langle v, \psi(y)-x\rangle}\\
&=
\sum_{\alpha\in \F_p^s}\widehat f(\alpha) \omega^{Q_\alpha(y)}\, \Exp_{v\in \F_p^k}\omega^{\langle v,\, (\gamma_\alpha + \psi)(y) - x\rangle}.
\end{align*}
Taking~$y = E(x) + Z$, we then obtain
\begin{align*}
\Pr\big[\phi\big(E(x) + Z\big) = x\big]
&=
\Exp_{x,Z}\mathbf{1}\big[\phi(E(x) + Z\big) = x\big]\\
&\leq
\sum_{\alpha\in \F_p^s}|\widehat f(\alpha)|\, \Exp_{x,Z}\big| \Exp_{v\in \F_p^k}\omega^{\langle v,\, (\gamma_\alpha + \psi) (E(x) + Z) - x \rangle}\big|\\
&\leq
\bigg( \sum_{\alpha\in \F_p^s}|\widehat f(\alpha)| \bigg) \max_{\alpha\in \F_p^s} \Exp_{x,Z} \mathbf{1}\big[(\gamma_\alpha + \psi)\big(E(x) + Z\big) = x\big]\\
&\leq
p^{s/2} \max_{\alpha\in \F_p^s} \Pr\big[(\gamma_\alpha + \psi)\big(E(x) + Z\big) = x\big],
\end{align*}
where we have used the Cauchy-Schwarz inequality and Parseval's identity in the last line.
Since~$\deg(\gamma_{\alpha} + \psi) \leq d-1$ and~$s$ ultimately depends only on~$p$, $d$, $\rho$ and~$\varepsilon$, by taking
\beqn
k \geq k_0(p,d,\rho,\eps) := k_0(p, d-1, \rho, \eps\, p^{-s/2})
\eeqn
we conclude from the induction hypothesis that
\beqn
\Pr\big[\phi\big(E(x) + Z\big) = x\big] \leq \varepsilon
\eeqn
in this case as well.
The theorem follows.

\section{Quantum circuit for decoding the Hadamard code}
\label{sec:qcircuit}
In this section we give the constant-depth quantum circuit to decode the Hadamard code.
We first give the operations on a high level, then we present more details on the implementations and finally we show how to further reduce the total complexity.

\subsection{High-level quantum algorithm}
Interestingly, we can implement each of the operations performed by the entangled players described in Section~\ref{sec:qstrategy} with constant-depth quantum circuits using only single and two-qubit gates and classical parity gates, thus giving a \qnc$^0[\oplus]$ circuit. 

To generate GHZ states (which we will need on multiple occasions), we use a technique of Watts et al.~\cite{WattsKothariSchaefferTal:2019} that starts by generating a so-called \emph{poor man's cat state}: $\frac{1}{\sqrt{2}}(\ket{z}+\ket{\bar{z}})$ for some binary vector $z$. 
To generate an $n$-qubit poor man's cat state, we apply Hadamard gates to an $n$-qubit register initialized in the all-zeros state and then compute the parity between adjacent qubits in $n-1$ auxilliary qubits. 
After measuring the auxilliary qubits, we are left with a poor man's cat state. 
With the parity measurements $d_i$ we can correct the poor man's cat state to a GHZ state by flipping qubits conditioned on a prefix-sum computation of the measurement outcomes. 

To apply the conditional phase-flip~\eqref{eq:condphase}, we use an auxilliary qubit and compute the AND function 
\beqn
\ket{z}\ket{b}\quad\mapsto\quad \ket{z}\ket{\text{AND}(z_1,\hdots,z_k)\oplus b}.
\eeqn
We then apply a phase-flip on the last qubit conditioned on~$c(y)=1$.
Using $X$-gates, we can ensure that AND evaluates to $1$ if and only if $z=y$. 
For ease of implementation, we use the identity AND$(z_1,\hdots,z_k)=\neg \text{OR}(\neg z_1,\hdots,\neg z_k)$. 

We cannot use standard decomposition techniques, such as using Toffoli gates, to implement the OR, as these do not give constant-depth circuits. 
Instead, we use the constant-depth Exact OR-implementation of~\cite{TakahashiTani:2013}. Their method uses single and two-qubit gates and additionally makes use of quantum fan-out gates that implement the general map
\beqn
\ket{b}\ket{y_1}\hdots\ket{y_n}\mapsto\ket{b}\ket{y_1\oplus b}\hdots\ket{y_n\oplus b}.
\eeqn 
Below, we show how to implement these quantum fan-out gates with \qnc$^0[\oplus]$ circuits. 
These quantum fan-out gates compute parity when conjugated with Hadamard gates on each input.
We use the quantum fan-out gates to compute and sum the parity of each subset of the inputs to compute the OR of the inputs in the first qubit~\cite[Lemma 1]{TakahashiTani:2013}. 

Implementing this constant-depth approach does come at the cost of a circuit size exponential in $k$: $O(k2^k)=O(n\log n)$.
We can implement this for each input in parallel which gives a quantum circuit of size $O(n^2\log n)$. 
As we have a constant number of constant-depth operations, we have constant-depth circuit for list-decoding the Hadamard code. 

\subsection{Details of quantum algorithm}
Now, we present details of the operations given in the previous section. 

Figure~\ref{fig:q_circuit:poor_man_cat_state} shows the quantum circuit to generate a 3-qubit GHZ state. 
If the measurement results were $d_1=1$ and $d_2=1$,  then we have the state $\frac{1}{\sqrt{2}}(\ket{010}+\ket{101})$. 
We flip the second qubit, because the first measurement result, $d_1$, equals~1. We do not flip the third qubit, because the parity of the measurement results, $d_1\oplus d_2$, equals~0.
This indeed gives the 3-qubit GHZ state $\frac{1}{\sqrt{2}}(\ket{000}+\ket{111})$ as desired.
\begin{figure}[ht]
	\centering
\begin{quantikz}
\lstick{$\ket{0}_1$} & \gate{H} & \ctrl{1} & \qw & \qw & \qw & \qw & \qw \\
\lstick{$\ket{0}_2$} & \qw & \targ{} & \targ{} & \meter{} & \gate[1,cwires={1}, style={white}]{d_1} & \cwbend{1} \\
\lstick{$\ket{0}_3$} & \gate{H} & \ctrl{1} & \ctrl{-1} & \qw & \qw & \targ{} & \qw \\
\lstick{$\ket{0}_4$} & \qw & \targ{} & \targ{} & \meter{} & \gate[1,cwires={1}, style={white}]{d_1 \oplus d_2} & \cwbend{1} \\
\lstick{$\ket{0}_5$} & \gate{H} & \qw & \ctrl{-1} & \qw & \qw & \targ{} & \qw  
\end{quantikz}
	\caption{The quantum circuit to generate a 3-qubit GHZ state. First we obtain a poor man's cat state $\frac{1}{\sqrt{2}}(\ket{z}+\ket{\bar{z}})$ with each $z\in\F_2^3$ equally likely to be found. The parity gates compute a prefix sum on the measurement results $d_1$ and $d_2$ and determine if a qubit has to be flipped to obtain the GHZ state.}
	\label{fig:q_circuit:poor_man_cat_state}
\end{figure}
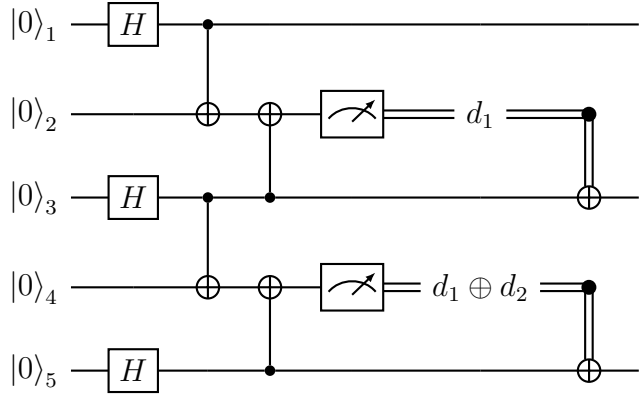

This method extends naturally to larger $n$. A prefix-sum computation is then used to determine which qubits have to be flipped. 
The depth of the circuit does not increase with larger $n$.
In our GHZ state construction, we implicitly assumed the qubits to be arranged in a linear architecture. 
Other arrangements work equally well, as shown in~\cite{WattsKothariSchaefferTal:2019}. 

A powerful tool in the implementation above is the quantum fan-out gate. We now show how to implement this gate using only single and two-qubit gates and classical parity gates. 
For this, we combine the GHZ state construction introduced above with ideas from distributed quantum computing, specifically, the non-local CNOT-gate~\cite{Eisert:2000,YimsiriwattanaLomonaco:2004}. 
The term non-local CNOT originates from the fact that we can imagine the control and target qubits being hosted on different quantum devices which share a GHZ state. 
We apply this gate in a local setting to construct the quantum fan-out gate in constant-depth. 

Figure~\ref{fig:q_circuit:non_local_cnot} shows the quantum fan-out gate implementation for one control and two targets. 
This method extends to an arbitrary number of targets by reapplying the same operations to all target qubits and the corresponding qubits in the GHZ state in parallel. 
The last $Z$-gate is only applied if the parity over all measurement results equals $1$. 
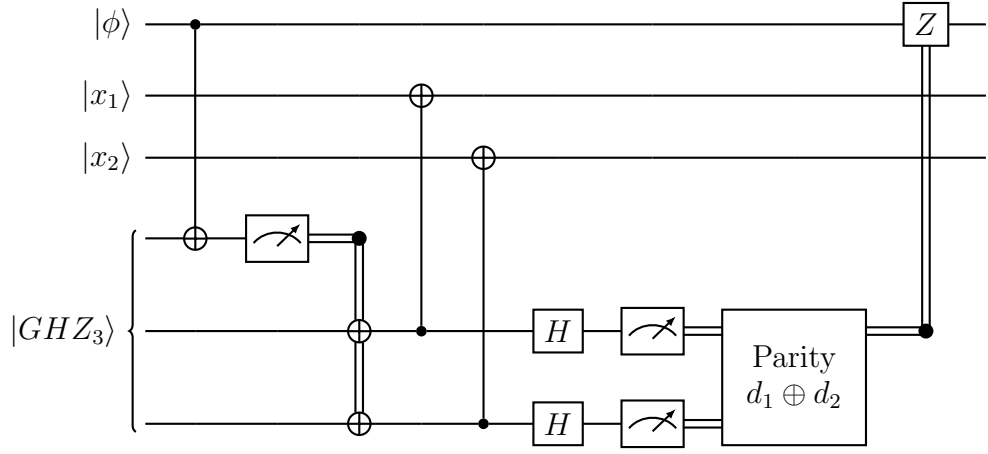
\begin{figure}
\centering
\begin{quantikz}
\lstick{$\ket{\phi}$} & \ctrl{3} & \qw & \qw & \qw & \qw & \qw & \qw & \qw & \gate{Z} & \qw  \\
\lstick{$\ket{x_1}$} & \qw & \qw & \qw & \targ{} & \qw & \qw & \qw & \qw & \qw & \qw \\
\lstick{$\ket{x_2}$} &  \qw & \qw & \qw & \qw & \targ{} & \qw & \qw & \qw & \qw & \qw \\
\lstick[wires=3]{$\ket{GHZ_3}$} & \targ{} & \meter{} & \cwbend{1} \\
 & \qw & \qw & \targ{} & \ctrl{-3} & \qw & \gate{H} & \meter{} & \gate[2,cwires={1,2}, disable auto height]{\begin{array}{c}\text{Parity} \\ d_1\oplus d_2\end{array}} & \cwbend{-4} \\
 & \qw & \qw & \targ{}\vcw{-1} & \qw & \ctrl{-3} &  \gate{H} & \meter{} & &  
\end{quantikz}
	\caption{Implementation of a quantum fan-out gate with one control qubit $\ket{\phi}$ and two target qubits $\ket{x_1}$ and $\ket{x_2}$. Only single and two-qubit gates and classical parity gates are used. The bottom three qubits are in the GHZ$_3$ state.}
	\label{fig:q_circuit:non_local_cnot}
\end{figure}

\begin{lemma}
The circuit of Figure~\ref{fig:q_circuit:non_local_cnot} implements a quantum fan-out gate. 
\end{lemma}
\begin{proof}
Let $\ket{x}$ be an $n$-qubit computational basis state and $\ket{\phi}=\alpha\ket{0}+\beta\ket{1}$ be any single qubit quantum state. 
We will prove that the circuit implements the quantum fan-out gate on the state $\ket{\phi}\ket{x}$. 
The lemma then follows by linearity of the operations. 

The action of the quantum fan-out gate on the quantum state is given by 
\begin{equation}
    \ket{\phi}\ket{x} 
    \overset{\text{fan-out}}{\mapsto} \alpha \ket{0} \ket{x} + \beta \ket{1} X^{\otimes n} \ket{x} = \alpha \ket{0} \ket{x} + \beta \ket{1} \ket{\bar{x}},
    \label{eq:action_fanout}
\end{equation}
where $\ket{\bar{x}}$ is the computational basis state $\ket{x}$ with all qubits flipped. 

To see why this works, assume we have a GHZ$_{n+1}$ state and we apply the operations as shown in Figure~\ref{fig:q_circuit:non_local_cnot} generalized to arbitrary $n$. 
Up to a normalization factor of $1/\sqrt{2}$ coming from the GHZ state we then have:
\begin{align*}
     \big[\alpha & \ket{0} + \beta \ket{1}\big] \ket{x} \otimes\big[\ket{00 \cdots 0} + \ket{11 \cdots 1}\big]  \\
    & \xmapsto{(1)} \alpha \ket{0}\ket{x}\otimes\big[\ket{00 \cdots 0} + \ket{11 \cdots 1}\big]   + \beta \ket{1} \ket{x}\otimes\big[\ket{10 \cdots 0} + \ket{01 \cdots 1}\big]  \\
    & \xmapsto{(2)} \alpha \ket{0}\ket{x}\ket{d_0 0 \cdots 0}  + \beta \ket{1} \ket{x}\ket{d_01 \cdots 1} \\
    & \xmapsto{(3)} \alpha \ket{0}\ket{x}\ket{d_0 0 \cdots 0}  + \beta \ket{1} X^{\otimes n}\ket{x}\ket{d_01 \cdots 1} \\
    & \xmapsto{(4)} \frac{1}{2^{n-1}} \sum_{d\in\F_2^n} \big[\alpha\ket{0}\ket{x} + \beta (-1)^{d_1 + \hdots + d_n}\ket{1} X^{\otimes n}\ket{x}\big]\ket{d_0 d_1\hdots d_n} \\
    & \xmapsto{(5)} \alpha\ket{0}\ket{x}\ket{d_0 d_1\hdots d_n}  + (-1)^{d_1 + \hdots + d_n}\beta \ket{1} X^{\otimes n}\ket{x}\ket{d_0 d_1\hdots d_n} \\
    & \xmapsto{(6)} \big[\alpha\ket{0}\ket{x} + \beta\ket{1}\ket{\bar{x}}\big]\ket{d_0d_1\hdots d_n}.
\end{align*}
In Step (1), we perform a CNOT operation from the control qubit to the first qubit of the GHZ state. 
In Step (2), we measure that qubit, with outcome $d_0$, and apply an $X$-gate to the remaining $n$ qubits of the GHZ state if $d_0=1$. 
Next we perform CNOT gates between the $i+1$-th qubit of the GHZ state and the $i$-th target qubit. 
In Steps (4) and (5), we first apply Hadamard gates to each unmeasured qubit of the GHZ state and subsequently measure it. 
Finally, we compute the parity $d_1\oplus\hdots\oplus d_n$ and apply a $Z$-gate to the control qubit if this parity equals one. 
This indeed gave the desired final state and hence the shown quantum circuit implements the quantum fan-out gate. 
\end{proof}

\subsection{Reducing the algorithm's complexity}
We now show how to reduce the complexity of the quantum algorithm from $O(n^2\log n)$ to $O(n\log n\log\log n)$.

With the current implementation, the complexity of applying a single conditional phase-flip is polynomial in the codeword length $n$. 
We can reduce this to polynomial in the message length $k$. 
For this we apply the OR-reduction~\cite{HoyerSpalek:2005}. 
Instead of evaluating an OR on $k$ inputs, we use a $O(k\log k)$ size constant depth circuit to prepare a quantum state on $\lceil\log(k+1)\rceil$ qubits, such that the OR on these $\lceil\log(k+1)\rceil$ qubits evaluates to the same value as the OR on the original $k$ qubits. 
This OR-reduction uses the quantum fan-out gate.

The OR-reduction increases the depth by an additive constant, however, it reduces the complexity of applying a single conditional phase-flip to $O(k\log k)$. 
The complexity of the quantum circuit therefore reduces to $O(n\log n\log\log n)$, using $k=\log n$. 

Note that the phase flip is only applied if the input is one, hence, we have to apply the OR-reduction and the Exact OR implementation only if this input is one.

\section{The high-characteristic setting}
\label{sec:high_char}

In this section we give some evidence to support our conjecture (made in Section~\ref{sec:cresults}) that the probability of correct message retrieval by \nc$^0[\oplus]$ circuits decays exponentially with the message length.
This is done by proving the following theorem, which a ``high-characteristic'' analogue of Theorem~\ref{thm:biased_equi} with much better bounds;
as we see no reason to believe a result of this kind has a strong dependence on the characteristic of the finite field considered,\footnote{While the result is stated in the setting of prime fields~$\F_p$, it easily generalizes to the case of non-prime finite fields $\F_q$, with only minor modifications in the proof.}
we believe that a similar bound also holds for low-characteristic fields such as~$\F_2$.

\begin{theorem}[Exponential decay in high characteristic] \label{thm:high_char}
For every $d\in \N$ and $\rho\in [0,1)$ there exist constants $C=C(\rho,d)$ and $c=c(\rho,d) > 0$ such that the following holds.
Let $p > d$ be a prime, and let $n$, $k$ be integers with $k\geq p$.
Then for every polynomial map $\phi:\F_p^n\to\F_p^k$ of degree at most~$d$ and every function $E:\F_p^k\to\F_p^n$ we have
\beqn
\Pr_{x\in \F_p^k, Z\sim \mathcal N_\rho(0)}\big[\phi\big(E(x) + Z\big) = x\big] \leq C e^{-c k/(\log k)^{d^2}}.
\eeqn
\end{theorem}

\begin{remark}
The presence of the poly-logarithmic term in the exponential above is due to a poly-logarithmic loss when passing between two distinct notions of tensor rank in our proof of Theorem~\ref{thm:high_char}.
It is a widely-believed conjecture in additive combinatorics that these two notions of rank (see Section~\ref{sec:tensors} below) are within a constant multiplicative factor of one another, in which case our proof would give an upper bound of the form $C e^{-c k}$ for the probability of correct message retrieval (and this would be the best possible).
\end{remark}

As with the proof of Theorem~\ref{thm:biased_equi}, we will prove Theorem~\ref{thm:high_char} by induction on the degree~$d$, using the degree-1 case shown in Section~\ref{sec:cstrategy} as the base case of the induction.
The inductive argument will also share many similarities with the one presented in Section~\ref{sec:classical}, in particular relying on a structure-versus-randomness dichotomy based on a notion of rank associated with the polynomial map~$\phi$.
The reason for the better bounds we obtain now stems from the fact that, in the high-characteristic case, one can work with tensors (i.e. multilinear forms) rather than with general polynomial maps.
In the quasirandom case of our argument, we can then use a stronger version of the random restriction theorem for the analytic rank of tensors (also proved in~\cite{BrietCS:2022}), while in the structured case we use a recently-proved close connection between analytic rank and partition rank of tensors~\cite{MoshkovitzZ:2022}.

\subsection{Tensors associated to polynomial maps}
\label{sec:tensors}

Given a polynomial map $\phi: \F_p^n \to \F_p^k$ of degree at most $d$, we can define a $(d+1)$-tensor $T: (\F_p^n)^d \times \F_p^k \to \F_p$ associated to it by
\beqn
T(y_1, \dots, y_d, v) = \big\langle v,\, \Delta_{y_1} \cdots \Delta_{y_d} \phi(0)\big\rangle,
\eeqn
where we recall that $\Delta_y \phi(x) = \phi(x+y) - \phi(x)$.
While not immediately obvious, the formula above indeed defines a tensor (i.e., it is linear in each variable separately).
This follows from the fact that $\Delta_{y_1} \cdots \Delta_{y_d} \phi$ does not depend on the order of the derivatives, and that the polynomial map $\Delta_{y_1} \cdots \Delta_{y_{d-1}} \phi$ has degree at most~1 (since~$\phi$ has degree at most~$d$);
note that, if~$\psi$ is a linear map, then~$h\mapsto\Delta_h\psi$ is linear in~$h$.

If the characteristic~$p$ of the field in strictly higher than the degree~$d$, then we also have the \emph{integration formula}
\beqn
\phi(y) = \frac{1}{d!} T(y, \dots, y, \cdot) + \psi(y) \quad \text{for all $y\in \F_p^n$,}
\eeqn
where $y$ is repeated $d$ times inside $T$ and $\psi$ is a polynomial map of degree at most~$d-1$.
This follows from the (discrete) Taylor expansion theorem, and allows us to pass back and forth between tensors and polynomial maps.

We will use the following two notions of rank for tensors, originally introduced by Gowers and Wolf~\cite{GowersW:2011} and by Naslund~\cite{Naslund2020}, respectively.

\begin{definition}[Tensor analytic rank]
Let $X_1,\dots,X_r$ be finite sets and~$T:\F_p^{X_1}\times\cdots\times\F_p^{X_r}\to \F_p$ be an $r$-tensor.
The bias of~$T$ is defined as
\beqn
\bias(T) = \Exp_{x_1\in \F_p^{X_1},\dots,x_r\in \F_p^{X_r}} \omega^{T(x_1,\dots,x_r)},
\eeqn
where $\omega = e^{2\pi i/p}$.
The bias is always real and positive,\footnote{It is not hard to show that $\bias(T) = \Pr_{x_1\in \F_p^{X_1},\dots,x_{r-1}\in \F_p^{X_{r-1}}} \big[T(x_1,\dots, x_{r-1}, \cdot) \equiv 0\big]$.}
and the analytic rank of~$T$ is defined by
\beqn
\arank(T) = -\log_{p} \bias(T).
\eeqn
\end{definition}

\begin{definition}[Partition rank]
A nonzero $r$-tensor $T:\F_p^{X_1}\times\cdots\times\F_p^{X_r}\to \F_p$ is said to have partition rank~1 if there is a nonempty strict subset $I\subset [r]$ and tensors $S: \prod_{i\in I}\F^{X_i}\to \F$ and $R: \prod_{i\in [r]\setminus I}\F^{X_i}\to \F$ such that~$T$ can be factored as $T = SR$ .
The partition rank of~$T$, denoted $\prank(T)$, is defined as the least $m\in \N$ such that there is a decomposition $T = T_1 + \cdots + T_m$ where each $T_i$ has partition rank~1.
\end{definition}

While these two notion of rank are defined in very different ways, it turns out that they are intimately related to one another.
Lovett has shown that $\arank(T) \leq \prank(T)$ holds for all tensors~\cite{Lovett2019}, and it is a well-known open problem to determine whether a similar inequality holds in the converse direction, up to an absolute multiplicative factor.
Very recently, Moshkovitz and Zhu~\cite{MoshkovitzZ:2022} proved that the relation between these two rank functions is at worst quasilinear.

\begin{theorem}[Moshkovitz--Zhu]\label{thm:MZ}
For every $r\geq 2$ there exists $L_r>0$ such that for every $r$-tensor~$T$ over any finite field, we have
\beq\label{eq:MZ}
\arank(T) \leq \prank(T) \leq L_r\arank(T)\,\log^{r-1}\big(1 + \arank(T)\big).
\eeq
\end{theorem}

This result will be an important ingredient in our proof of Theorem~\ref{thm:high_char};
we note that the decay obtained could be improved to $C e^{-ck}$ if Theorem~\ref{thm:MZ} were proven without the poly-logarithmic factor on the right-hand side of~\eqref{eq:MZ}.
Another important ingredient is the following random restriction theorem for tensors~\cite{BrietCS:2022}, stated here for the special case of the analytic rank.

\begin{theorem}[Tensor random restriction theorem]\label{thm:tensor_restriction}
For every $d\in \N$ and $\sigma \in (0,1]$, there exist constants $C,\kappa>0$ such that for any order-$d$ tensor~$T$ over any field, we have that
\beqn
\Pr_{I\sim [n]_\sigma}\big[\arank(T_{|I}) \geq \kappa\cdot \arank(T)\big]
\geq
1 - Ce^{-\kappa\,\arank(T)}.
\eeqn
\end{theorem}

\subsection{The proof of Theorem~\ref{thm:high_char}}

The proof will proceed by induction on the degree of the polynomial map.
Recall that in the base case of degree-1 maps the result has already been proven in Section~\ref{sec:cstrategy}.

Let now $\phi: \F_p^n \to \F_p^k$ be a polynomial map of degree at most $d$, with $2 \leq d < p$, and suppose the theorem holds for polynomial maps of degree at most $d-1$.
Define the $(d+1)$-tensor $T: (\F_p^n)^d \times \F_p^k \to \F_p$ by
\beqn
T(y_1, \dots, y_d, v) = \big\langle v,\, \Delta_{y_1} \dots \Delta_{y_d} \phi(0)\big\rangle.
\eeqn
We split the analysis into two cases, depending on whether the analytic rank of~$T$ is above or below some cut-off value $r = \Theta\big(k/(\log k)^{d^2}\big)$.

\subsubsection*{Pseudorandom case}

Assume that $\arank(T) \geq r$.
We will show that, for any given $x\in \F_p^n$, the probability
\beq \label{eq:pseudo_case}
\Pr_{Z\sim \mathcal N_\rho(0)}\big[\phi\big(E(x) + Z\big) = x\big]
\eeq
decays exponentially on $\arank(T)$;
we then conclude the pseudorandom case by averaging over all $x$.

Fix some $x\in \F_p^n$.
As before (in Section~\ref{sec:classical}), we write
\begin{align*}
    \Pr_{Z\sim \mathcal N_\rho(0)}\big[\phi\big(E(x) + Z\big) = x\big]
    &= \Exp_{I\sim [n]_{1-\rho}} \Pr_{y\in \F_p^I}\big[\phi\big(E(x) + y\big) = x\big] \\
    &= \Exp_{I\sim [n]_{1-\rho}} \Exp_{y\in \F_p^I, v\in \F_p^k} \omega^{\langle v,\, \phi(E(x)+y) - x\rangle}.
\end{align*}
Note that we can write $\phi(E(x)+y) - x = \phi(y) + \psi(y)$, where
\beqn
\psi(y) := \Delta_{E(x)}\phi(y) - x
\eeqn
has degree at most $d-1$.
Using this identity and the triangle inequality, it follows that
\begin{align*}
    \Pr_{Z\sim \mathcal N_\rho(0)}\big[\phi\big(E(x) + Z\big) = x\big]
    &= \Exp_{I\sim [n]_{1-\rho}} \Exp_{y\in \F_p^I, v\in \F_p^k} \omega^{\langle v,\, \phi(y) + \psi(y)\rangle} \\
    &\leq \Exp_{I\sim [n]_{1-\rho}}\Exp_{v\in \F_p^k} \big| \Exp_{y\in \F_p^I} \omega^{\langle v,\, (\phi + \psi)(y) \rangle} \big|.
\end{align*}

Repeated applications of the Cauchy-Schwarz inequality (or equivalently, the monotonicity property of the Gowers uniformity norms~\cite[pp.~420]{tao_vu_2006}) shows that, for any fixed $v\in \F_p^k$, $I\subseteq [n]$, we have
\begin{align*}
    \big| \Exp_{y\in \F_p^I} \omega^{\langle v,\, (\phi + \psi)(y) \rangle} \big|
    \leq \big( \Exp_{y_0, y_1, \dots, y_d\in \F_p^I} \omega^{\langle v,\, \Delta_{y_1} \dots \Delta_{y_d}(\phi + \psi)(y_0) \rangle} \big)^{1/2^d}.
\end{align*}
We then conclude that
\begin{align*}
    \Pr_{Z\sim \mathcal N_\rho(0)} &\big[\phi\big(E(x) + Z\big) = x\big] \\
    &\leq \Exp_{I\sim [n]_{1-\rho}}\Exp_{v\in \F_p^k} \big( \Exp_{y_0, y_1, \dots, y_d\in \F_p^I} \omega^{\langle v,\, \Delta_{y_1} \dots \Delta_{y_d}(\phi + \psi)(y_0) \rangle} \big)^{1/2^d} \\
    &\leq \Exp_{I\sim [n]_{1-\rho}} \big(\Exp_{v\in \F_p^k} \Exp_{y_0, y_1, \dots, y_d\in \F_p^I} \omega^{\langle v,\, \Delta_{y_1} \dots \Delta_{y_d}(\phi + \psi)(y_0) \rangle} \big)^{1/2^d},
\end{align*}
where we have applied H\"{o}lder's inequality once
(or, alternatively, Cauchy-Schwarz $d$ further times).

Now we need to relate this last expression to the analytic rank of~$T$.
Deriving~$d$ times a polynomial map of degree at most $d-1$ gives the zero map, and so
$\Delta_{y_1} \dots \Delta_{y_d}\psi(y_0) \equiv 0$.
Moreover, since $\deg(\phi) \leq d$, the $d$-th derivative $\Delta_{y_1} \dots \Delta_{y_d}\phi$ is a constant map.
We conclude that
\beqn
\Exp_{v\in \F_p^k} \Exp_{y_0, y_1, \dots, y_d\in \F_p^I} \omega^{\langle v,\, \Delta_{y_1} \dots \Delta_{y_d}(\phi + \psi)(y_0) \rangle}
= \Exp_{v\in \F_p^k} \Exp_{y_1, \dots, y_d\in \F_p^I} \omega^{\langle v,\, \Delta_{y_1} \dots \Delta_{y_d}\phi(0) \rangle}.
\eeqn
For each $v\in \F_p^k$, let $S(v)$ be the $d$-tensor given by $T(\cdot,\dots,\cdot,v)$.
Then the above is precisely the bias of the restricted tensor $S(v)_{|I^d}$, averaged over~$v$, which (by definition) equals the average of $p^{-\arank(S(v)_{|I^d})}$.
The probability~\eqref{eq:pseudo_case} is then bounded from above by
\beqn
\Exp_{v\in \F_p^k}\Exp_{I\sim [n]_{1-\rho}} p^{-\arank(S(v)_{|I^d})/2^d}.
\eeqn
Theorem~\ref{thm:tensor_restriction} now implies that for some absolute constant $C = C(d,\rho)>0$, the last quantity is bounded from above by $Cp^{-\arank(T)/C}$.
This settles the pseudorandom case.

\subsubsection*{Structured case}
Now we assume that $\arank(T) < r$.

Denote the partition rank of $T$ by $s := \prank(T)$.
Theorem~\ref{thm:MZ} shows that $s \leq L_{d+1} r (\log r)^d$, where $L_{d+1}$ is a universal constant.
We can then write
\beqn
T(y_{[d]}, v) = \sum_{i=1}^s R_i(y_{I_i}) S_i(y_{I_i^c}, v)
\eeqn
for some non-empty sets $I_i \subseteq [d]$, $|I_i|$-tensors $R_i$ and $(d-|I_i|+1)$-tensors $S_i$.
Since $d<p$, by Taylor's expansion theorem we have that
\beqn
\phi(y) = \frac{1}{d!} \Delta_y \dots \Delta_y \phi(0) + \psi_0(y), \quad \deg(\psi_0) < d.
\eeqn
Define $q_i: \F_p^n \to \F$, $\psi_i: \F_p^n \to \F_p^k$ ($i\in [s]$) by
\beqn
q_i(y) = \frac{1}{d!} R_i(y^{I_i}), \quad \langle v, \psi_i(y)\rangle = S_i(y^{I_i^c}, v),
\eeqn
and note that $\deg(\psi_i) < d$ for all $i\in [s]$.
By the definition of $T$ we conclude that
\beqn
\phi(y) = \psi_0(y) + \sum_{i=1}^s q_i(y) \psi_i(y), \quad \text{with $\deg(\psi_i) < d$ for $0\leq i\leq s$.}
\eeqn

Let $\mathcal{A} = \{A_1,\dots,A_m\}$ be the partition of $\F_p^n$ given by the level sets of the polynomial map $(q_1,\dots,q_s):\F_p^n\to\F_p^s$;
note that $m \leq p^s \leq p^{L_{d+1} r (\log r)^d}$.
For each $j\in [m]$, $\phi$ will coincide on $A_j$ with a polynomial map $\psi_{A_j}: \F_p^n \to \F_p^k$ of degree at most $d-1$
(just substitute the $q_i(y)$ on the formula above by their value on $A_j \in \mathcal{A}$).
Define the random events
\beqn
\mathcal{E}_j = \big\{ E(x)+Z \in A_j:\, x\sim \mathcal{U}(\F_p^k),\, Z\sim \mathcal{N}_\rho(0)\big\}, \quad j\in [m].
\eeqn
Since these events partition the probability space, it follows that
\begin{align*}
    \Pr_{x\in \F_p^k, Z\sim \mathcal N_\rho(0)} &\big[\phi\big(E(x) + Z\big) = x\big] \\
    &= \sum_{i=1}^m \Pr_{x, Z}\big[\phi\big(E(x) + Z\big) = x \And \mathcal{E}_j\big] \\
    &= \sum_{i=1}^m \Pr_{x, Z}\big[\psi_{A_j}\big(E(x) + Z\big) = x \And \mathcal{E}_j\big] \\
    &\leq m \cdot \max_{1\leq j\leq m} \Pr_{x, Z}\big[\psi_{A_j}\big(E(x) + Z\big) = x\big] \\
    &\leq p^{L_{d+1} r (\log r)^d} \cdot \max_{\deg(\psi)<d} \Pr_{x, Z}\big[\psi\big(E(x) + Z\big) = x\big],
\end{align*}
where the last maximum is over all polynomial maps $\psi: \F_p^n \to \F_p^k$ of degree at most $d-1$.
By the induction hypothesis we have that this maximum probability is at most $C' e^{-c' k/(\log k)^{(d-1)^2}}$, where $C' = C(d-1, \rho)$ and $c' = c(d-1, \rho)$;
we conclude that
\begin{align*}
    \Pr_{x\in \F_p^k, Z\sim \mathcal N_\rho(0)} &\big[\phi\big(E(x) + Z\big) = x\big] \\
    &\leq C'\exp\bigg((\log p) L_{d+1} r (\log r)^d - \frac{c'k}{(\log k)^{(d-1)^2}}\bigg).
\end{align*}

Taking
\beqn
r = \frac{c'}{2L_{d+1}} \frac{k}{(\log k)^{d^2}},
\eeqn
and using our assumptions $k \geq p$ and $d\geq 2$, we have that
\begin{align*}
    (\log p) L_{d+1} r (\log r)^d
    &\leq (\log k) L_{d+1} \frac{c'}{2L_{d+1}} \frac{k}{(\log k)^{d^2}} (\log k)^d \\
    &= \frac{c'}{2} \frac{k}{(\log k)^{d^2-d-1}} \\
    &\leq \frac{c'}{2} \frac{k}{(\log k)^{(d-1)^2}}.
\end{align*}
We conclude that
\beqn
    \Pr_{x\in \F_p^k, Z\sim \mathcal N_\rho(0)} \big[\phi\big(E(x) + Z\big) = x\big]
    \leq C'\exp\bigg(-\frac{c'k}{2(\log k)^{(d-1)^2}}\bigg)
\eeqn
in this case, and the theorem follows.

\appendix
\section{MAJORITY from list decoding}
\label{sec:sudan}

This appendix shows how to compute the MAJORITY function when given oracle access to circuits capable of list decoding the Hadamard code.

\subsection{Classical circuits}

We start by considering the case of classical circuits, in particular proving Theorem~\ref{thm:sudan}, which we recall below for convenience.
Our proof of this result follows the arguments exposed in \cite[Section~6.2]{viola:2006}.

\ThmSudanMajority*

Let $\maj_t$ denote the MAJORITY function on $t$ bits.
We first introduce a promise problem called $\isbal_t$, which asks to determine whether a given binary string is balanced.
We then show that a (possibly probabilistic) circuit that solves $\isbal_t$ can be turned into a deterministic circuit that computes $\maj_t$.
Finally, we show how a circuit for $\LH_n(\eps)$ can be used to solve $\isbal_t$ for $t =\Omega(1/\eps)$.

\begin{definition}[The $\isbal_t$ problem]\label{def:isbal}
For an even positive integer~$t$, define $\isbal_t:\{x\in \F_2^t :\, |x|\leq t/2\}\to \F_2$ by
\beqn
\isbal_t(x) = \left\{
\begin{array}{ll}
	1 &\text{if $|x| = t/2$}\\
	0 &\text{otherwise.}
\end{array}
\right.
\eeqn
Given an arbitrary~$x\in \F_2^t$, define the $\isbal_t$ problem to be to return $\isbal_t(x)$ if $|x| \leq t/2$ and an arbitrary bit otherwise.
\end{definition}

\begin{lemma}[Derandomization lemma]\label{lem:prob_IsBal}
Let~$\mathcal{C}$ be a probabilistic circuit that solves $\isbal_t$ with probability at least~$2/3$ for every input.
There exists a deterministic oracle \ac$^0$ circuit~$\mathcal{C'}$ that, when given oracle access to~$\mathcal{C}$ and the ability to fix its random bits, solves~$\isbal_t$.
\end{lemma}

\begin{proof}
For some large enough constant~$c\in \N$, consider~$ct$ parallel instances of~$\mathcal{C}$.
It follows from the Chernoff bound that, for any fixed $x\in \F_2^t$ given to all of these instances, with probability $1- \exp(-10\, t)$ at least 55\% of the instances solves the $\isbal_t$ problem on input~$x$.

By the union bound, one can fix the randomness in the instances of~$\mathcal{C}$ in order to get a deterministic classical circuit that, for every input~$x\in\F_2^t$ with $|x| \leq t/2$, returns a $ct$-bit string whose Hamming weight is at least $0.55 t$ if $\isbal_t(x) = 1$ and at most $0.45 t$ if $\isbal_t(x) = 0$.
Distinguishing these two types of strings is known as the approximate majority problem, for which there is an \ac$^0$ circuit~\cite{Ajtai:1983}.
Combining these circuits gives the result.
\end{proof}

We now show that a deterministic circuit that solves $\isbal_t$ can be used to compute $\maj_t$. 

\begin{lemma}\label{lem:det_IsBal}
Let~$\mathcal{C}$ be a deterministic circuit for $\isbal_t$.
There exists an oracle \ac$^0$ circuit~$\mathcal{D}$ that, given oracle access to~$\mathcal{C}$, computes $\maj_t$.
\end{lemma}

\begin{proof}
For $x\in\F_2^t$ and $i\in \{0,1,\dots,t\}$, define $x_i$ as the string $x$ with the first~$i$ bits set to zero and the rest of the bits equal to those of~$x$. 
So, for instance, $x_0 = x$ and $x_t$ is the all-zeroes string. 
Let $\mathcal{D}$ be the circuit that runs $t+1$ parallel instances of $\mathcal{C}$ with inputs $x_0,x_1,\dots,x_t$, respectively, and returns the OR of the $t+1$ outputs.

We claim that~$\mathcal{D}$ computes $\maj_t$.
Indeed, if $x$ has fewer than $t/2$ ones then~$\mathcal{C}$ returns $0$ for each input $x_i$, as the number of 1s only decreases with $i$. 
If $x$ has at least~$t/2$ ones, then $\mathcal{C}$ returns $1$ for at least one $i$, since $x_0$ has at least $t/2$ ones, whereas $x_t$ is the all-zeroes string. 
This completes the proof.
\end{proof}

Towards turning a circuit $\mathcal{C}$ for $\LH_n(\eps)$ into a circuit for $\isbal_t$, we associate with each input $x\in \F_2^t$ to $\isbal_t$ a random error vector~$N_x$ over~$\F_2^n$ as follows:
independently, each coordinate of~$N_x$ is a uniformly random entry of~$x$.
In particular, for balanced $x$, the error vector $N_x$ will correspond to an error rate of $1/2$ and we refer to it as $N_{1/2}$. 
The next lemma shows that there is a message $m\in\F_2^k$ that has small probability of recovery by $\mathcal{C}$ under the error vector $N_{1/2}$.

\begin{lemma}\label{lem:decoding_random_message}
Let $\mathcal{C}$ be a probabilistic circuit which, on input $y\in \F_2^n$, returns a (random) list $L(y)\subseteq \F_2^k$ of at most~$2^k/4$ elements.
Then there exists $m\in \F_2^k$ such that
\begin{equation}\label{eq:upper_bound_random_decoding}
\Pr[m\in L(H(m) + N_{1/2})] \le 1/4,
\end{equation}
where the probability is taken over~$L$ and~$N_{1/2}$.
\end{lemma}

\begin{proof}
Note that, for any $y\in \F_2^n$, the vector $y+N_{1/2}$ is uniformly distributed over~$\F_2^n$;
in particular, it has the same distribution as~$N_{1/2}$.
Let $M\in \F_2^k$ be a uniformly distributed random element.
Then, by independence of $M$, $L(y)$ and~$N_{1/2}$, get that
\begin{align*}
\Pr_{M,L,N_{1/2}}[M\in L(H(M) + N_{1/2})] & = \Pr_{M, L, N_{1/2}}[M\in L(N_{1/2})] \\
& \leq \frac{1}{2^k}\, \Exp_{L,N_{1/2}}\, |L(N_{1/2})|  \\
& \leq 1/4.
\end{align*}
Hence, there exists a value~$m$ of~$M$ such that~\eqref{eq:upper_bound_random_decoding} holds. 
\end{proof}

Finally, we prove that the circuit $\mathcal{C}$ in Theorem~\ref{thm:sudan} can solve~$\isbal_t$. 

\begin{lemma}\label{lem:list_decoding_implies_IsBal}
Let $\mathcal{C}$ be a probabilistic circuit as in Theorem~\ref{thm:sudan}. 
There exists a probabilistic oracle \ac$^0$ circuit $\mathcal{D}$ of size $\poly(n,1/\eps)$ that, when given oracle access to $\mathcal{C}$, solves $\isbal_t$ with probability at least~3/4 for $t = \Omega(1/\eps)$. 
\end{lemma}

\begin{proof}
We may assume without loss of generality that $\eps \leq 1/4$.
Let $\delta \in [\eps, 1/4]$ be minimized such that $t = 1/(2\delta)$ is an even integer;
note that, since $\delta \geq \eps$, the circuit~$\mathcal{C}$ also solves $\LH_n(\delta)$ with probability at least~$3/4$.
Fix a message $m$ as in Lemma~\ref{lem:decoding_random_message}, and let $x\in \F_2^t$ be any given string (which serves as input to~$\mathcal{D}$).

The circuit~$\mathcal{D}$ has three layers.
The first layer has the string~$H(m)$ hardwired into it and uses~$n$ independent uniform samples to the coordinates of~$x$ to compute the random string $H(m) + N_x$.
This layer is a probabilistic circuit using~$n$ parallel two-bit XOR gates.
The second layer consists of the circuit~$\mathcal{C}$, which produces a random list $L(H(m)+ N_x)$ of size at most~$n/4$.
The third layer consists of an~\ac$^0$ circuit of size $\poly(n)$ that returns $0$ if and only if $m \in L(H(m)+ N_x)$. 
This can be done by checking equality between~$m$ and the~$O(n)$ elements of the list.
We claim that this solves $\isbal_t$.

If~$x$ is balanced then it follows from Lemma~\ref{lem:decoding_random_message} that~$\mathcal{D}$ correctly returns~$1$ with probability at least~$3/4$. 
If~$x$ has Hamming weight strictly less than~$t/2$, then each coordinate of~$N_x$ is~1 with probability at most $1/2 - 1/t = 1/2 - 2\delta$.
By the Chernoff bound,~$N_x$ has Hamming weight at most $(1/2 - \delta)n$ with probability $1- \exp(-\Omega(\delta^2 n))$.
Hence, the properties of the circuit~$\mathcal{C}$ imply that in this case~$\mathcal{D}$ correctly outputs $0$ with probability at least $3/4$. 
\end{proof}

Theorem~\ref{thm:sudan} now follows directly by combining Lemma~\ref{lem:prob_IsBal}, Lemma~\ref{lem:det_IsBal} and Lemma~\ref{lem:list_decoding_implies_IsBal}.

\subsection{The quantum case}
\label{sec:Sudan_quantum}

Now we sketch how the above proof can be used to turn our \qnc$^0[\oplus]$ circuit for decoding the Hadamard code into one that computes MAJORITY with polynomially small error.
We first recall the following result.

\qlisthad*

Let $\eps = n^{-1/4}$ and let~$\mathcal{C}$ be the circuit from Corollary~\ref{cor:qlisthad}.
Since~$\mathcal{C}$ returns lists of size at most~$n^{3/4}$, a stronger version of Lemma~\ref{lem:decoding_random_message} holds where the probability~\eqref{eq:upper_bound_random_decoding} -- taken additionally over the measurement outcomes of~$\mathcal{C}$ -- is bounded from above by~$n^{3/4}/2^k = n^{-1/4}$.

The proof of Lemma~\ref{lem:list_decoding_implies_IsBal} then gives an oracle \qnc$^0[\oplus]$ circuit~$\mathcal{D}$ of size~$\poly(n)$ that, given oracle access to~$\mathcal{C}$, solves $\isbal_t$ with probability $1 - O(n^{-1/4})$ for $t = \Omega(n^{1/4})$.
Here, the \ac$^0$ circuit used to check membership of~$m$ can be replaced with our \qnc$^0[\oplus]$ circuit for the OR function (see Section~\ref{sec:qcircuit}) applied to the entrywise sum of~$m$ with each element in the list.

Now let $t' = \lfloor n^{1/8} \rfloor$, and note that the same circuit $\mathcal{D}$ above can be used to solve $\isbal_{t'}$ with probability $1 - O(n^{-1/4})$:
it suffices to pad the input with zeroes and ones in the same number until we have a string of the correct size.
Finally, with the proof of Lemma~\ref{lem:det_IsBal} and the union bound we obtain a \qnc$^0[\oplus]$ circuit~$\mathcal{D'}$ that, given oracle access to~$\mathcal{D}$, solves $\maj_{t'}$ with probability $1 - O(n^{-1/8})$ for $t'  = \lfloor n^{1/8} \rfloor$.
Here again we use our \qnc$^0[\oplus]$ circuit for the OR function as explained in Section~\ref{sec:qcircuit}.

\section*{Acknowledgements}
We thank Richard Cleve for helpful discussions on the proof of Theorem~\ref{thm:qcircuit}, Emanuele Viola for helpful pointers to the literature and comments on an earlier version of this manuscript, Tamar Ziegler for encouragement to write up Section~\ref{sec:high_char}, as well as anonymous referees for helpful comments.

\bibliographystyle{alphaurl}
\bibliography{hgame}

\end{document}